\documentclass[letterpaper, 10 pt, conference]{ieeeconf} 
\usepackage{booktabs}
\usepackage{blindtext, graphicx}
\usepackage{cite}
\usepackage{graphicx}
\usepackage{amsmath}
\usepackage{amsfonts}
\usepackage{color}
\usepackage{enumerate}
\usepackage{subcaption}
\usepackage[export]{adjustbox}
\usepackage{caption}
\usepackage{amssymb}   
\usepackage{comment}
\usepackage[ruled,vlined]{algorithm2e}
\usepackage{multicol}

\usepackage{color,soul}
\usepackage[english]{babel}

\usepackage[utf8]{inputenc}
\usepackage[english]{babel}

\usepackage[utf8]{inputenc}

\usepackage{graphicx}
\usepackage{amsmath}
\usepackage[version=4]{mhchem}
\usepackage{siunitx}
\usepackage{longtable,tabularx}
\newtheorem{theorem}{Theorem}
\newtheorem{assumption}{Assumption}

\newtheorem{lemma}{Lemma}
\newtheorem{definition}{Definition}
\newtheorem{remark}{Remark}

\DeclareUnicodeCharacter{2212}{-}

\IEEEoverridecommandlockouts                             
\title{\LARGE \bf
	Cyber-Attack Detection in Discrete Nonlinear Multi-Agent Systems \\ Using Neural Networks
}

\author{Amirreza~Mousavi, Kiarash~Aryankia, and Rastko~R.~Selmic,~\IEEEmembership{Senior~Member,~IEEE}% <-this % stops a space
	\thanks{A. Mousavi, K. Aryankia and R. R. Selmic are with the Department of Electrical and Computer Engineering,
		Concordia University, Montreal, QC, Canada.
		{\tt\small se\_usavi@encs.concordia.ca,     k\_aryank@encs.concordia.ca,
			\tt\small rastko.selmic@concordia.ca.}}}

\begin{document}
	\maketitle
	\thispagestyle{empty}
	\pagestyle{empty}

	\begin{abstract}
		This paper proposes a distributed cyber-attack detection method in communication channels for a class of discrete, nonlinear, heterogeneous, multi-agent systems that are controlled by our proposed formation-based controller. A residual-based detection system, exploiting a neural network (NN)-based observer, is developed to detect false data injection attacks on agents' communication channels. A Lyapunov function is used to derive the NN weights tuning law and the attack detectability threshold. The uniform ultimate boundedness (UUB) of the detector residual and formation error is proven based on the Lyapunov stability theory. The proposed method’s attack detectability properties are analyzed, and simulation results demonstrate the proposed detection methodology’s performance.

	\end{abstract}

	\section{Introduction}\label{section:1}

	Cyber-Physical Systems (CPSs) are the integration of computation units and communication networks with physical processes \cite{lee2008cyber}.  In recent years, much attention has been devoted to studying CPSs due to their modern engineering applications such as traffic networks\cite{xiong2015cyber},\cite{chen2017cyber}, smart grids \cite{macana2011survey},\cite{pasqualetti2013attack}, Internet of things, and autonomous multi-agent systems such as unmanned aerial vehicles (UAVs) and unmanned ground vehicles (UGV)  \cite{baheti2011cyber}.  Most of the aforementioned systems are connected to the Internet and wireless communication networks through communication channels that attackers can penetrate and change the transmitted data. Several cyber-attacks have been reported in recent years \cite{kriaa2012modeling}, \cite{case2016analysis} which can deteriorate physical systems' performance and ultimately lead to failures or unsafe behaviour. As a result, significant attention has been devoted to the study of the security of CPSs. Various cyber-attack detection methods have been proposed in the literature. In \cite{miao2016coding}, a sensor coding mechanism is used to detect stealthy data injection attacks, which is designed by an intelligent attacker with a system model knowledge. In \cite{boem2017distributed}, the problem of detecting cyber-attacks on the communication network between interconnected subsystems, governed by a consensus-based control, is investigated and a distributed residual-based attack detection method is proposed for detecting attacks on the neighbouring communication channels. In \cite{farivar2019artificial}, a strategy is proposed to estimate and compensate attacks in the forward link of a nonlinear CPS. The proposed method is using nonlinear control theory with applied neural networks to develop cyber-attack observers for multi-agents. In \cite{jin2017adaptive}, an adaptive framework is developed for the control design of cyber-physical systems in the presence of simultaneous adversarial sensor and actuator attacks. In \cite{ding2016observer}, an event-trigger consensus control for stochastic linear discrete-time multi-agent systems with lossy sensors and cyber-attacks is designed to guarantee the prescribed consensus, and a distributed observer has been developed to estimate the relative full states.

	In multi-agent systems, consensus and formation are two essential problems that researchers study. In consensus control, agents interact locally in order to reach a common value of a certain state \cite{olfati2007consensus}. The formation is defined as a configuration in a space, where each agent is at the desired distance or angle from its neighbours \cite{xiao2009finite}. In this paper, attack detection problem is addressed for formation of a discrete, nonlinear multi-agent system. In a formation control of multi-agent systems, it is necessary for agents to communicate with each other to achieve the formation objective \cite{oh2015survey}. However, these channels are vulnerable to cyber-attacks. By changing the channel data, agents receive the corrupted data. Receiving the attacked data by each agent violates the formation and increases the possibility of a collision in the system. Therefore, the security of the received data is of paramount importance. %Moreover, in the nonlinear systems, any small change in the communication channel data can significantly impact the agents' states and make them unstable. 

	In multi-agent systems, each agent has three types of communication channels: (i) actuator channel, which transfers the control signal from the controller to the plant; (ii) sensor channel, which transfers the system output (the sensor measurements) from the agent plant to the controller; and (iii) neighbouring channels through which each agent receives neighbours' data. These vulnerable communication channels are prone to cyber-physical attacks\cite{teixeira2010networked,boem2017distributed}.

	Some distributed methods for attack detection in multi-agent systems have been recently proposed \cite{khazraei2017replay}, \cite{pang2016two}, \cite{barboni2020detection},\cite{arrichiello2015observer}, \cite{huang2019reliable}, \cite{wu2017secure}. Most of the works for multi-agent systems consider linear or continuous models without the leader. Moreover, some methods assume that each agent knows the entire topology of the multi-agent system and requires the global model's knowledge.

	%\begin{figure}[h]
	%    \centering
	%    \includegraphics[scale=0.7]{image/system.png}
	%   \caption{Agent system under attack.}
	%    \label{fig:universe}
	%\end{figure}  

	On the other hand, in this paper, we propose a method that requires knowledge of the local agent's model and locally available information or information communicated by neighbouring agents. We present cyber-attack detection methods in a discrete first-order nonlinear multi-agent system with unknown dynamics using an NN-based observer by using radial basis functions neural network (RBFNN). The proposed method is distributed, where the detection action in each agent relies on the agent's local information and the received data from its neighbours. Moreover, through Lyapunov stability analysis, a threshold is obtained for the proposed residual-based detector to detect agents' communication channels' attacks.

	The attack model that has been studied in this paper is a false data injection (FDI) attack on the actuator channel, sensor channel and neighbouring channels. We assume that the attacker does not have access to the agents' system dynamics. As a result, the covert attack cannot be applied to the system.

	The main contributions of this paper are as follows:
	\begin{enumerate}
		\item Development of an observer-based attack detection scheme as well as the detectability condition for a class of discrete, nonlinear multi-agent systems with unknown dynamics.
		\item Development of a distributed NN-based attack detection system and a NN-based controller for formation of discrete, nonlinear multi-agent systems with unknown dynamics.
		\item Demonstration of how the proposed system is capable of detecting attacks on actuator, sensor and neighbouring channels of multi-agent systems.
		
	\end{enumerate}

	This paper is organized as follows: in Section \ref{section:2}, the problem formulation and the attacks models are presented. Then, in Section \ref{section:3}, the NN-based observer and controller are developed. In Section \ref{section:4} the attack detectability condition is given, and in Section \ref{section:5} the performance of proposed cyber-attack detection system is demonstrated through the simulation results.

	\section{Preliminaries and Problem Formulation}\label{section:2}

	\subsection{Graph Theory}
	Interaction among the agents can be modelled by directed graph. This paper considers the weighted directed graph to define these interactions and communication graph which represents the information flow. Multi-agent systems can be modelled with a weighted directed graph\cite{mesbahi2010graph}. Let directed graph $\mathcal{D}$ be given by $\mathcal{D} = (\mathcal{V} ,\mathcal{E})$ where the set of vertices is $\mathcal{V} = \left\{ {{\upsilon _1},{\upsilon _2},...,{\upsilon _N}} \right\}$, and the set of edges is $\mathcal{E} \subseteq \mathcal{V} \times \mathcal{V}$. An edge in $\mathcal{D}$ is denoted by an ordered pair $(\upsilon_j, \upsilon_i)$ which is rooted at node $j$ and ended at node $i$. In the multi-agent system, the node $\upsilon _i$ denotes the \textit{i}-th agent, and  $(\upsilon_j, \upsilon_i) \in \mathcal{E}$, if and only if, agent $j$ sends information to agent $i$. The node $j$ is called a neighbour of node $i$ if $(\upsilon_j, \upsilon_i) \in \mathcal{E}$. The direction of an arrow in the directed graph matches the direction of information flow. 
	
	The adjacency matrix for a weighted directed graph is defined as $\mathcal{A} = [{a_{ij}}]$  where  $a_{ij} > 0$ when $(\upsilon_j, \upsilon_i) \in \mathcal{E}$, otherwise $a_{ij} = 0$. we consider, there is no self-loops e.g., $a_{ii} = 0$. The element of $a_{ij}$ is the weight between vertices $\upsilon _i$ and $\upsilon _j$. The $ d_{in}(\upsilon _i)$ is sum of in-degree weight to vertex $\upsilon _i$ which is defined as follows:
	\begin{equation}\label{eq:1}
	d_{in}(\upsilon _i) = \sum_{\{j|({\upsilon _j}{\upsilon _i}) \in \mathcal{E} \}} a_{ij}.
	\end{equation}
	The Degree matrix $\Delta \in \mathbb{R}^{N\times N}$ which is defined as $\Delta = diag(d_{in}(\upsilon _i)$), and in-degree weighted Laplacian is defined as $L = \Delta - \mathcal{A}$.
	
	\begin{definition}\label{def:1}
		The neighbour set of $i${-}th agent, \textbf{$\mathcal{N}_i = \{j|(\upsilon_j, \upsilon_i)\in \mathcal{E} \}$} is the set of agents that send their information to agent $i$.
	\end{definition}

	A  direct path from node $i$ to node $j$ is a sequence of successive edges in the form $\{(\upsilon_i, \upsilon_m), (\upsilon_m, \upsilon_l), . . . , (\upsilon_k, \upsilon_j)\}$. We use $\upsilon_l$ to annotate the leader of the multi-agent system.  The communication topology between the $N$ follower agents is assumed to be a directed graph with a fixed topology. The digraph is strongly connected, i.e. there is a directed path from $\upsilon_i$ to $\upsilon_j$ for all distinct nodes. Moreover, we assume that there is at least one directed path from the leader to one of the agents, i.e. the leader-follower structure contains a spanning tree where the leader has a directed path to all of the followers.

	In this paper, $x^{+}$ represents the state at time $t+1$, $tr(.)$ denotes the trace of matrix, and $\underline{1} \in \mathbb{R}^N$ is the vector of $1$'s. Moreover, $\bar{\sigma}(.)$ denotes the maximum singular value of a matrix and $\underline{\sigma}(.)$ denotes the minimum singular value of a matrix.  For any vector $x$ the notation $||x||$ denotes Euclidean norm and the Frobenius norm of any matrix $A$ is $||A||_F=\sqrt{tr(A^T A)}$.

	\subsection{System Model Without Attack}
	We consider a nonlinear discrete-time multi-agent system that consists of  $N$ agents. The dynamics of each agent is given by 
	\begin{equation} \label{eq:2}
	x^{+}_i = f_i(x_i) + u_i + w_i,\\ 
	\end{equation} 
	where $x_i \in \mathbb{R}^n$ is the system state, $u_i \in \mathbb{R}^n $ is the control input, and  $w \in \mathbb{R}^n$ is the structural disturbance. Functions $f_i(.): \mathbb{R}^n \rightarrow \mathbb{R}^n $ is locally Lipschitz nonlinear function. We also assume that the dynamics nonlinearity $f_i$ and disturbance $w_i$ are unknown. The overall system dynamics can be written as
	% In this paper we assume that the full state of each agent is measurable by the sensors, and the system output $y_i  = x_i $ which is sent to the neighbouring agents of agent $i$. But for simplicity we use $x_i$ instead of $y_i$. 
	\begin{equation} \label{eq:3}
	\mathbf{x}^{+} = {f}(\mathbf{x}) + {u} + {w},\\ 
	\end{equation} 
	where the stacked state vector is $\mathbf{x}=[x^T_1, ...,x^T_N]^T \in \mathbb{R}^{nN}$, ${f}(\mathbf{x})= [f^T_1(x_1),...,f^T_N(x_N)]^T: \mathbb{R}^{nN} \rightarrow \mathbb{R}^{nN}$, control input  ${u} = [u^T_1, ...,u^T_N]^T\in \mathbb{R}^{nN}$, and ${w}=[w^T_1, ...,w^T_N]^T \in \mathbb{R}^{nN}$.

	\begin{assumption}\label{assumptio:1}
		The unknown disturbance $w$ is bounded by $||w|| \leq w_M$ with $w_M$ a fixed bound. 
	\end{assumption}
	
	The leader dynamics is defined as follows:
	\begin{equation} \label{eq:4}
	x^{+}_l = f_l(x_l),
	\end{equation}   
	where $f_l \in \mathbb{R}^n$, and should satisfy the following assumption.
	\begin{assumption}\label{assumptio:2}
		The leader dynamics $f_l(x_l )$ is bounded by $||f_l(x_l )|| \leq F_M$, with a fixed bound $F_M$.
	\end{assumption}
	
	\vspace{0.1 cm}
	The objective of each agent is to reach a desired relative position with its neighbouring agents. We designed a distributed control law is designed such that the relative position between neighbouring agents $i$ and $j$ converges to the bounded desired relative inter-agent displacement $d_{ij}$:
	\begin{equation}\label{eq:5}
	x_i − x_j \rightarrow d_{ij}. ~~~ i,j= 1, . . . , N
	\end{equation}
	
	We define the desired relative inter-agent displacement between agent $i$ and the leader as $d_i$, as a result we can represent $d_{ij} = d_i - d_j$, and define the tracking error between the agent $i$ and leader as
	\begin{equation}\label{eq:6}
	\delta_i  = x_l  - x_i  - d_i.
	\end{equation}
	
	%\begin{remark}
	%Here, we do not use $\delta_i$, because it is a global quantity that cannot be computed locally at each node. As such, it is suitable for analysis but not for distributed controls design using Lyapunov techniques. \hl{ Vague}
	%\end{remark}

	%\noindent where $d_i \in  \mathbb{R}^p$ is a constant vector. Then, by defining $d_{ij}  = d_i − d_j$ the desired relative positions between agents can be defined as 

	The local formation error of $i$-th agent is defined as:
	\begin{equation}\label{eq:7}
	e_i  = \sum_{j\in \mathcal{I}_i}a_{ij}(x_j -x_i  - d_{ji}) + b_i(x_l -x_i -d_i),
	\end{equation}
	where $b_i$ is the direct gain from agent $i$ to the leader and $b_i \geq 0$, with $b_i > 0$ for at least one agent. Then $b_i \neq 0$ only if node $i$ can directly observe the state of the leader. Let matrix $B =diag(b_i)$. The global form of formation error is given as
	\begin{equation}\label{eq:8}
	e  = -[(L+B)\otimes I_n](\mathbf{x}  - \underline{1} \otimes {x}_l  - d),
	\end{equation}
	where $d = [d^T_1,...,d^T_N]^T \in \mathbb{R}^{nN}$. Since $b_i \neq 0$ for at least one agent, matrix $L+ B$ is full-rank and invertible. The desired relative position between agents and the leader $d$ is bounded by $||d|| < d_M$ with $d_M$ a fixed bound.
	
	By defining the stacked tracking error as $\delta = [\delta^T_1,..., \delta^T_N]^T$ the global synchronisation error can be also written as
	\begin{equation}\label{eq:9}
	e  = -[(L+B)\otimes I_n]\delta , 
	\end{equation}
	and the global synchronization error dynamics can be written as 
	\begin{equation}\label{eq:9-1}
	e^{+} = -[(L+B)\otimes I_n](f(\textbf{x}) + u + w - \underline{1} \otimes f(x_l) - d).
	\end{equation}
	\begin{lemma}[\hspace{-1.4mm}~\cite{das2011cooperative}]\label{lem11}
		Let the graph is strongly connected and $B \neq 0$. Then  $||\delta|| \leq ||e||/\underline{\sigma}(L+B)$.
	\end{lemma}

	\subsection{Radial Basis Function Neural Network}  
	A NN is used to approximate the unknown nonlinearity  $f_i(x_i)$ in equation (\ref{eq:2}) over a compact set $\Omega_i$ which can be written as: 
	\begin{equation}\label{eq:10}
	{f_i}({x}_i ) = {W}_i^{T}{ \varphi_i}({x}_i ) + \epsilon_i ,
	\end{equation}
	
	\noindent where the $W_i \in \mathbb{R}^{\vartheta_i \times n}$ is the desired constant unknown weight matrix of the NN, ${ \varphi_i}({x}_i)\in \mathbb{R}^{\vartheta_i \times n}$ is a NN activation function, and $\epsilon_i$ is the NN approximation error. Additional details on NNs can be found in \cite{lewis2020neural}. The approximation of the overall dynamics nonlinearity $f(\mathbf{x})$ can be written as
	\begin{equation}\label{eq:11}
	{f}({\mathbf{x} }) = {W}^{T}{ \varphi}(\mathbf{x} ) + \epsilon ,
	\end{equation}
	where the overall ideal NN weight matrix $W$ is defined as $W = diag(W_1,...,W_N)$, $\varphi(\mathbf{x}) = [\varphi_1^T(x_1),...,\varphi_N^T(x_N)]^T$, and $\epsilon = [\epsilon_1^T,...,\epsilon_N^T]^T$.
	
	The estimation of the nonlinearity in the dynamics of multi-agent systems using RBFNN is given as
	\begin{equation}\label{eq:12}
	\hat{f}_i({x}_i ) = \hat{W}_i^{T} { \varphi_i}({x}_i),
	\end{equation}

	\noindent where $\hat{W}_i$ is the current estimated NN weight matrix, and the used activation function $\varphi(x_i )$ is given by
	\begin{equation}\label{eq:13}
	\varphi(x_i) = exp[\frac{-(x_i-m_i)^T(x_i-m_i)}{p_i}],
	\end{equation}
	where $m_i$ is the center of the activation function, and $p_i$ is the width of the Gaussian function. The overall nonlinearity $f(\mathbf{x})$ over a compact set $\Omega$  can be estimated as follows:
	\begin{equation}\label{eq:14}
	\hat{f}(\mathbf{x} ) = \hat{W}^{T} { \varphi}(\mathbf{x} ),
	\end{equation}
	where the estimation of the ideal weight matrix $\hat{W}$ is defined as $\hat{W} = diag(\hat{W}_1,...,\hat{W}_N)$. We use some standard assumptions \cite{cui2015distributed,das2011cooperative, aryankia2020formation} as follows :
	\vspace{0.1 cm}
	\begin{assumption}\label{assumptio:3}
		Unknown ideal NN weight matrix $W$ is bounded by $||W||_F \leq W_M$, with a fixed bound $W_M$.
	\end{assumption}
	\vspace{0.1 cm}
	%\begin{assumption}
	%NN activation functions $\varphi$ are bounded for each $i$, so $ ||\varphi(x )|| \leq  \varphi_M$.
	%\end{assumption}
	\begin{assumption}\label{assumptio:4}
		The NN approximation error $\epsilon $ is bounded by $||\epsilon || \leq \epsilon_M$ with a fixed bound $\epsilon_M$.
	\end{assumption}

	\begin{figure}[tb]\label{fig:1}
		\centering
		\includegraphics[scale=.55]{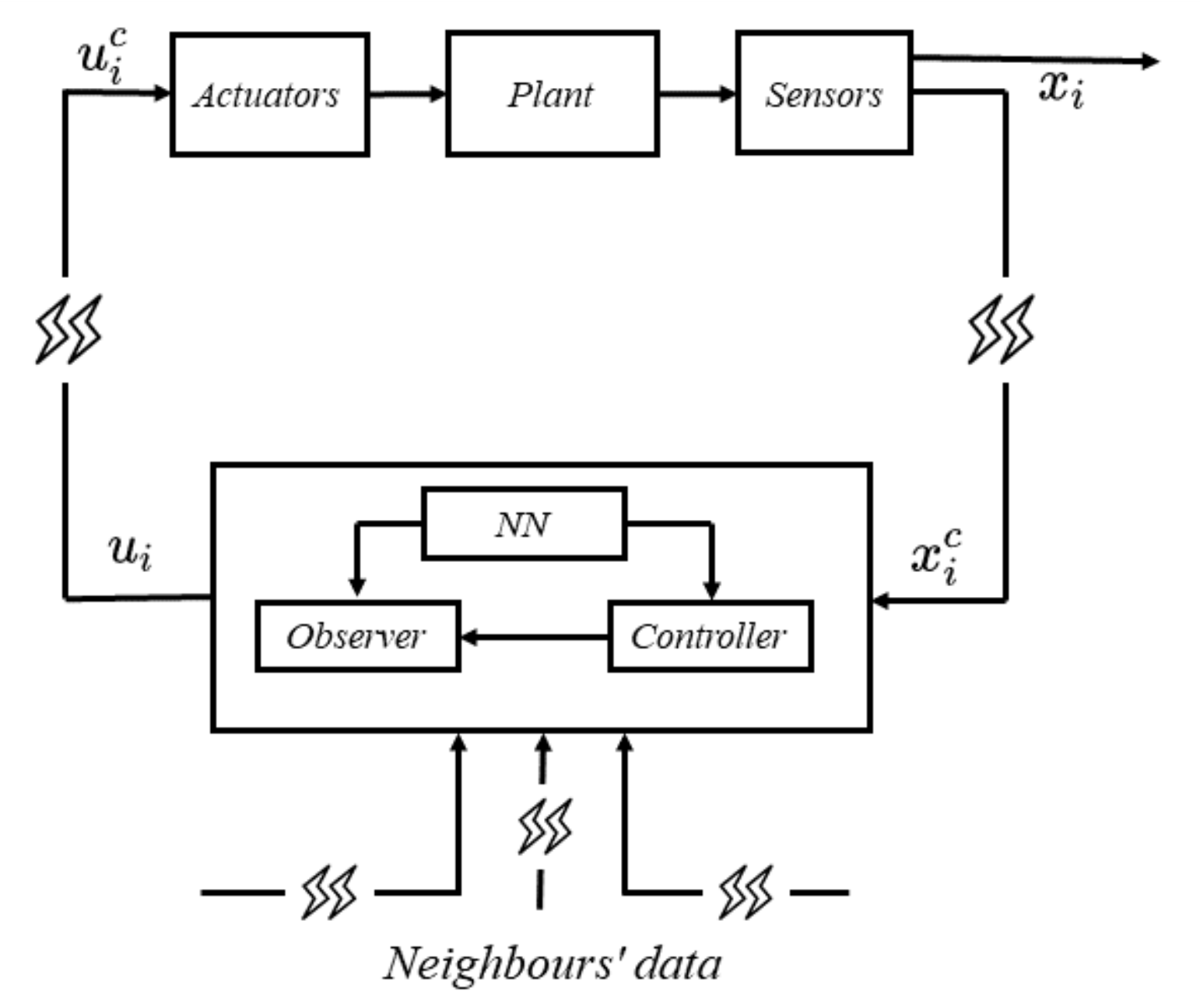}
		\caption{Agent system architecture.}
		\label{fig:universe}
	\end{figure}

	\subsection{Attacks Model}% note: check attack model 
	We assume that the attacker has access to the agent communication channels, and it can change the actuator, sensor and neighbouring channels data of the agents. %In the following we model each of the on the multi-agent system.
	
	\textbf{Attack on actuator channel:} Each agent uses the actuator channel to send its control input to the plant, and the attacker can perform the attack on this channel and change the control input. The attack on the actuator channel can be modelled as   
	\begin{equation}\label{eq:15}
	u^c_i  = u_i  + \kappa_i  u^a_i ,
	\end{equation}
	where $u^c_i $ is the corrupted control input, and $\kappa_i$ is "0" in the attack-free case, and it is "1" when there is an attack in the actuator channel.
	
	\textbf{Attack on sensor channel:} The system output is sent to the controller through the sensor channel, and the attack can change the sensor data by injecting some false data into the sensor channel. We assume that the system's state is measurable by sensor \cite{modares2019resilient}. As a result, the attack on the sensor channel can be modelled as follows:
	\begin{equation}\label{eq:16}
	x^c_i  = x_i  + \lambda_i  x^a_i ,
	\end{equation}
	where $x^c_i $ is the corrupted sensor data, and $\lambda_i$ is $"0"$ when there is not an attack on sensor channel, otherwise it is $"1"$. So the attacker by corrupting the sensor channel and changing its data, affects the control input since the controller uses the sensor data.

	\textbf{Attack on neighbouring channel:} We consider that each agent's control input is a function of the agent's sensor data and its neighbouring data. If we define $\zeta_i$ as aggregated outputs of $i${-}th agent's neighbours, then the $i${-}th agent's control input is %a function of its states and $\zeta_i$, as follows:
	\begin{equation}\label{eq:17}
	u_i  = \mathcal{F}_i(x_i ,\zeta_i ).
	\end{equation}   
	
	Let vector $x^a_j $ denote the injected attack signal into the neighbouring channel of agent $i$, where $j \in \mathcal{N}_i$. The neighbouring channel between agent $i$ and $j$ in presence of attack can be modelled as \cite{boem2017distributed}:
	\begin{equation}\label{eq:18}
	\bar{x}^c_j  = \bar{x}_j  + \phi^i_j  \bar{x}^a_j ,
	\end{equation}
	where $\bar{x}^c_j $ is the corrupted neighbouring data, and $\phi^i_j $ is $"1"$ when there is an attack on neighbouring channel, otherwise it is $"0"$.
	\begin{remark}
		The $\bar{x}_j $ is the sensor data that agent $j$  ($j \in \mathcal{N}_i$) sends to agent $i$ and it is equal to $x_j$.
	\end{remark}
	\vspace{0.2cm}
	
	When the attacker, by injecting $\bar{x}^a_j $, changes $\bar{x}_j$ to $\bar{x}^c_j$, the $\zeta_i$ is changed to ${\zeta}_i^c$. Therefore, we have the following expression for the corrupted control input of agent $i$
	\begin{equation}\label{eq:19}
	{u}_i^{\prime} =\mathcal{F}_i(x_i ,{\zeta}_i^c ),
	\end{equation}
	
	\noindent where ${u}_i^{\prime}$ is the $i$-th agent's control input, which has been affected by the attack on the agent's neighbouring channels.
	
	\begin{assumption}\label{assumptio:5}
		The attacks can compromise the communication channels between followers, but the leader's communication channel is safe, and the leader can send its information without getting distorted by any attack.
	\end{assumption}

	%\noindent  where ${x}_i^a  \in \mathbb{R}$ is the system corrupted states, and ${u}_i^a  \in \mathbb{R} $ is the corrupted control inputs. Based on the system model (\ref{eq:1}), when the agent $i$ receives the attacked data from its neighbours, it generates corrupted control inputs ${u}_i^a $ and based on (\ref{eq:8}) the system output get affected (${x}_i^a$). As a result, the attack (corrupted data) propagates through the multi-agent system. %What is essential here, is to identify the source of the attack. When the source of the attack is identified, we can locate the attacked channel.Therefor, we can effectively compensate the effect of attack or at least mitigate its impact.
	%\vspace{0.1 cm}

	All types of attacks compromise the agent's control input, which can lead to degrading the control performance, distorting the formation and increasing the possibility of collision between agents.

	We propose a distributed residual-based attack detector to detect the aforementioned types of attacks. The attack detection system's objective is to enable each agent to detect attacks in its communication channels.  By exploiting the proposed detector, each agent can detect attacks on its sensor channel, actuator channel, and its neighbouring channels. 
	
	\section{Main Result}\label{section:3}
	
	For a distributed cyber-attack detection system, each agent has a dedicated NN-based observer that generates the residual signal. We propose the following observer to estimate the $i$-th agent’s states
	\begin{equation}\label{eq:20}
	\begin{split}
	\hat{x}^{+}_i = &\hat{W}_i^{T}\varphi_i({x}_i) + u_i - G_i(x_i - \hat{x}_i)\\ & - \sum_{j\in \mathcal{N}_i}a_{ij}(x_j-x_i - d_{ji}) + b_i(x_l-x_i-d_i),
	\end{split}
	\end{equation}
	where the diagonal matrix with nonnegative elements $G_i \in \mathbb{R}^{n \times n}$ is the observer gain.

	%As mentioned before, each agent's objective is to keep a desired relative position with its neighbouring agents and minimize the formation error (\ref{eq:8}). 
	
	Moreover, based on the defined control objective, we propose the following distributed control
	\begin{equation}\label{eq:21}
	u_i  = -\hat{f}_i(x_i ) +  c(x_i + k_ie_i)  ,
	\end{equation}
	with the diagonal matrix with nonnegative elements $k_i\in\mathbb{R}^{n \times n}$ and scalar $ c >0$ are the control gains. By using the equation (\ref{eq:12}), the control law (\ref{eq:21}) can be rewritten as 
	\begin{equation}\label{eq:22}
	u_i  = -\hat{W}_i^{T} { \varphi_i}({x}_i )  +  c(x_i+ k_i e_i) ,
	\end{equation}
	or in stacked form
	\begin{equation}\label{eq:23}
	u  = -\hat{W}^{T} { \varphi}(\mathbf{x} )  + c(\mathbf{x}+ Ke) ,
	\end{equation}
	where $K = diag(k_1,...,k_N)$.
	
	By defining $\tilde{W}_i = W_i − \hat{W}_i$, the function estimation error is
	\begin{equation}\label{eq:23-1}
	\tilde{f}(x) = f(x) − \hat{f}(x) = \tilde{W}^T \varphi(x) + \epsilon,
	\end{equation}
	with $\tilde{W}^T = diag(\tilde{W}_i)$. Let the observer error $\tilde{x}_i = x_i - \hat{x}_i$ be the attack detection residual. Then, one can obtain the observer error dynamics as 
	\color{black}
	%From equations (\ref{eq:2}), (\ref{eq:10}) and (\ref{eq:20}) the states estimation error is expressed as 
	\begin{equation}\label{eq:24}
	\Tilde{x}_i ^{+} = G_i\Tilde{x}_i + \Tilde{W}_i^{T} \varphi(x_i ) + \epsilon_i  + w_i + e_i.
	\end{equation}
	The stacked observer error is given by
	\begin{equation}\label{eq:25}
	\tilde{\mathbf{x}} ^{+} = G \tilde{\mathbf{x}}  + \tilde{W}^T \varphi(\mathbf{x} ) + w  +\epsilon  + e ,
	\end{equation}
	where $G = diag(G_1,...,G_N) \in \mathbb{R}^{nN\times nN}$.

	Let us consider the NN weight matrix tuning law as
	
	\begin{equation}\label{eq:26}
	\begin{split}
	\hat{W}_i ^{+} = &\hat{W}_i  + \alpha\varphi(x_i )\Bar{h}_i^T- F_i\hat{W}_i,
	\end{split}
	\end{equation}
	where scalar $\alpha>0$, $F_i = \gamma I_{\vartheta_i}$ with $0 < \gamma < 1$, and $I_{\vartheta_i }$ denotes $\vartheta_i \times \vartheta_i$ identity matrix, and
	\begin{equation}\label{eq:27}
	\small
	\Bar{h}_i = \Tilde{W}_i^{T} \varphi(x_i) + \epsilon_i  + w_i.
	\end{equation}
	
	%\begin{remark}
	%The parameters $\Bar{h}_i$ is obtainable from (\ref{eq:24}), by using $\Bar{h}_i = \tilde{x}_i^{+} - G_i\tilde{x}_i - e_i$.
	%\end{remark}
	
	%\textcolor{red}{
	%The following theorem shows that the mutlti-agent system with the proposed controller and NN observer maintains the formation and formation error is UUB.}
	
	\begin{definition} [\hspace{-1.4mm}~\cite{sarangapani2018neural}]
		Consider the following nonlinear system
		\begin{equation}
		x^{+} = F(x,t),
		\end{equation}
		where $x$ denotes the system state, and $F$ is a nonlinear function. Let the initial time be $t_0$, and the initial condition be $x(t_0) = x_0$.
		The solution is said to be uniformly ultimately bounded (UUB) if, for all initial states $x_0$, there exists a $b \in \mathbb{R}$ and an $\mathcal{T}_f(b,x_0) \in \mathbb{Z}^{+}$ such that $||x|| \leq b$ for all $t \ge t_0 + \mathcal{T}_f$. 
	\end{definition}

	\begin{theorem}\label{theorem:1}
		Consider a multi-agent system in the absence of attack to be modelled by a weighted undirected graph, for the class of nonlinear systems described by (\ref{eq:2}),(\ref{eq:3}), under the Assumptions 1-4, with control law (\ref{eq:22}), and the NN weights matrix tuning law as (\ref{eq:26}).  If the following conditions hold
		\begin{equation}\label{eq:28}
		0 < \bar{\sigma}(K) < \frac{1}{\bar{\sigma}(\bar{L})},
		\end{equation}
		\begin{equation}\label{eq:29}
		0 < c < \frac{1}{\sqrt{\eta{\bar\sigma}(P^TP)}},
		\end{equation}
		\begin{equation}\label{eq:30}
		\alpha < \frac{1}{||\varphi(x )||^2},
		\end{equation}
		with $\bar{L} = (L+ B) \otimes I_{n}$, $\eta = 1 + (1-\alpha\varphi^T\varphi)^{-1}$, and $P = I - K\bar{L}$, then the NN weight matrix estimation errors $\tilde{W}$ and formation error $e$ are UUB, with practical bounds given by
		(\ref{eq:62}) and (\ref{eq:63}) respectively.
	\end{theorem}
	
	\begin{proof}
		See Appendix.
	\end{proof}
	
	The system reaches its desired formation if no attack is injected into the system. We propose here a method to detect attacks. As a result, the following theorem is given.

	\begin{theorem}\label{theorem:2}
		
		Consider a multi-agent system that has reached the desired formation in the absence of an attack with the observer (\ref{eq:20}). If the following condition holds
		
		\begin{equation}\label{eq:31}
		\bar{\sigma}(G) < \frac{1}{\sqrt{\eta}},
		\end{equation}
		then, the observer error $\tilde{x}_i$ is UUB, with practical bounds given by (\ref{eq:43}).
	\end{theorem}
	\color{black}
	\begin{proof}
		Consider the desired formation is achieved if there is no injected attack to the multi-agent system ($||e||\le e_M$). Let us define the following Lyapunov candidate 
		\begin{equation}\label{eq:32}
		V = V_1 + V_2, 
		\end{equation}
		where  $V_1 =  \tilde{x}^T\tilde{x}$, and $V_2= \frac{1}{\alpha}tr(\tilde{W}^T\tilde{W})$. The first difference of Lyapunov candidate is given by
		\begin{equation}\label{eq:33}
		\Delta V = \Delta V_1 + \Delta V_2.
		\end{equation}
		Let us define $\mu  = w  + \epsilon $ and $\theta  = \tilde{W}^T \varphi(x )$. Based on Assumption 1 and Assumption 4, $\mu $ is bounded by
		\begin{equation}\label{eq:34}
		||\mu || \leq \mu_M,
		\end{equation}
		where $\mu_M = \epsilon_m + w_M$. One can derive $\Delta V_1$ by using (\ref{eq:25}) as follows:
		\begin{equation}\label{eq:35}
		\begin{split}
		\Delta V_1 =& ({\tilde{\mathbf{x}}}^{+})^T \tilde{\mathbf{x}} ^{+} - \tilde{\mathbf{x}}^T \tilde{\mathbf{x}} =
		-\tilde{\mathbf{x}}^T[I - G^TG]\tilde{\mathbf{x}}  \\&+2\tilde{\mathbf{x}}^TG^T\theta +2\theta^Te +2\tilde{\mathbf{x}}^TG^T\mu +2\tilde{x}G^Te\\& + \theta^T\theta+2\theta^T\mu+
		\mu^T\mu + e^Te + 2\mu^Te.
		\end{split}
		\end{equation}
		By defining $F = diag(F_i)$, we use (\ref{eq:26}) to obtain $\Delta V_2$ 
		\begin{equation}\label{eq:36}
		\small
		\begin{split}
		\Delta V_2 =&  \frac{1}{\alpha}tr[{(\tilde{W}^{+}})^T \tilde{W} ^{+}- \tilde{W}^T \tilde{W} ] = \\&
		\frac{1}{\alpha}tr[-2\alpha\tilde{W}^T\varphi\varphi^T\tilde{W}- 
		2\alpha\mu\varphi^T{\tilde{W}} +2\gamma\tilde{W}^T\hat{W}B
		\\&+\alpha^2\tilde{W}^T\varphi\varphi^T\varphi\varphi^T\tilde{W} +2\alpha^2\tilde{W}\varphi\varphi^T\varphi\mu^T\\&-2\alpha\tilde{W}\varphi\varphi^TF\hat{W}+\alpha^2\mu\varphi^T\varphi\mu^T-2\alpha\mu\varphi^TF\hat{W}\\&+\hat{W}^TF^TF\hat{W}].
		\end{split}
		\end{equation}

		Thus, from (\ref{eq:35}) and (\ref{eq:36}) one we can write $\Delta V$ as:
		\begin{equation}\label{eq:37}
		\small
		\begin{split}
		\Delta V =& 2\theta^TG\tilde{\mathbf{x}} - \tilde{\mathbf{x}}^T[I - G^TG]\tilde{\mathbf{x}} + 2\theta^Te+ 2\tilde{x}G^Te+2\theta^T\mu \\& +2\tilde{\mathbf{x}}^TG^T\mu+ 2\mu^Te+ \mu^T\mu+ e^Te +\\& \frac{1}{\alpha}tr [-2\alpha\tilde{W}^T\varphi\varphi^T\tilde{W}- 2\alpha\mu\varphi^T{\tilde{W}}  +2\gamma\tilde{W}^T\hat{W}B+\\&\alpha^2\tilde{W}^T\varphi\varphi^T\varphi\varphi^T\tilde{W} +2\alpha^2\tilde{W}\varphi\varphi^T\varphi\mu^T-2\alpha\tilde{W}\varphi\varphi^TF\hat{W}\\&+\alpha^2\mu\varphi^T\varphi\mu^T-2\alpha\mu\varphi^TF\hat{W}+\hat{W}^TF^TF\hat{W}],
		\end{split}
		\end{equation}

		%So, $\Delta V$ can be written as follows:
		%\begin{equation}
		%\small
		%    \begin{split}
		%        \Delta V =& -(1 - \alpha\varphi^T\varphi)\theta^T\theta - 2(1 - \alpha\varphi^T\varphi)\theta^T\mu + \\& 2\theta^T\mu+2\theta^TG\tilde{\mathbf{x}} - e^T[I-K^T\bar{L}^T\bar{L}K]e-\\&2e^TK^T\bar{L}^T\bar{L}\nu+\nu^T\bar{L}^T\bar{L}\nu - \tilde{\mathbf{x}}^T[I - G^TG - \\& G^T\bar{L}^T\bar{L}G - C^T\bar{L}^T\bar{L}C +2G^T\bar{L}^T\bar{L}C]\tilde{\mathbf{x}} + \\& 2\Tilde{\mathbf{x}}^T(G-C)^T\bar{L}^T\bar{L}\nu+2\tilde{\mathbf{x}}^TG^T\mu+\mu^T\mu \\&+\alpha\varphi^T\varphi\mu^T\mu + \frac{1}{\alpha}tr[2\tilde{W}^T(I - \alpha\varphi\varphi^T) F\hat{W}- \\&-2\alpha\mu\varphi^TF\hat{W}+\hat{W}^TF^TF\hat{W}]
		%   \end{split}
		%\end{equation}

		%\begin{equation}
		%\small
		%    \begin{split}
		%        \Delta V \leq & -(1 - \alpha\varphi^T\varphi)\theta^T\theta - 2(1 - \alpha\varphi^T\varphi)\theta^T\mu\\&+ 2\theta^T\mu+2\theta^TG\tilde{\mathbf{x}} - e^T[I-K^T\bar{L}^T\bar{L}K]e-\\&2e^TK^T\bar{L}^T\bar{L}\nu+\nu^T\bar{L}^T\bar{L}\nu - \tilde{\mathbf{x}}^T[I - G^TG - \\& G^T\bar{L}^T\bar{L}G - C^T\bar{L}^T\bar{L}C +2G^T\bar{L}^T\bar{L}C]\tilde{\mathbf{x}} + \\& 2\Tilde{\mathbf{x}}^T(G-C)^T\bar{L}^T\bar{L}\nu+2\tilde{\mathbf{x}}^TG^T\mu+\mu^T\mu \\&+\alpha\varphi^T\varphi\mu^T\mu + \frac{1}{\alpha}tr[ 2\gamma {\tilde{W}^T}(I-\alpha\varphi\varphi^T) \hat{W}\\&-2\alpha \mu{\varphi^T}\gamma {\hat{W}}+\gamma ^2\hat{W}^T\hat{W}]
		%    \end{split}
		%\end{equation}
		\noindent recognizing the terms in (\ref{eq:37}) yields
		\begin{equation}\label{eq:38}
		\small
		\begin{split}
		\Delta V \le & -(1 - \alpha\varphi^T\varphi)\theta^T\theta - 2(1 - \alpha\varphi^T\varphi)\theta^T\mu + 2\theta^Te\\& + 2\theta^T\mu+2\theta^TG\tilde{\mathbf{x}} - \tilde{\mathbf{x}}^T[I - G^TG]\tilde{\mathbf{x}} + 2\tilde{x}^TG^Te \\&+2\tilde{\mathbf{x}}^TG^T\mu+\mu^T\mu 
		+\alpha\varphi^T\varphi\mu^T\mu +e^Te + 2\mu^Te+ \\&\frac{1}{\alpha} tr[ \gamma ^2\hat{W}^T\hat{W} +2\gamma  {\tilde{W}}^T(W - \tilde{W})+2\alpha\gamma   \mu\varphi^T(\tilde{W} - W)].
		\end{split}
		\end{equation}
		Completing the squares for $\theta$, one can obtain
		\begin{equation}\label{eq:39}
		\small
		\begin{split}
		\Delta V \leq & -(1-\alpha\varphi^T\varphi)||\theta - \frac{1}{1-\alpha\varphi^T\varphi}G\tilde{\mathbf{x}}- \frac{\gamma +\alpha\varphi^T\varphi }{1-\alpha\varphi^T\varphi}\mu\\&-\frac{1}{1-\alpha\varphi^T\varphi}e||^2 - \tilde{\mathbf{x}}^T\Big[I - (1+ \frac{1}{1-\alpha\varphi^T\varphi})G^TG \Big]\tilde{\mathbf{x}}\\&  + 2(\frac{\gamma  +1}{1-\alpha\varphi^T\varphi})\tilde{x}^TG^T\mu + 2\eta\tilde{x}^TG^Te
		+(-2\gamma\\& + \frac{(1+\gamma  )^2}{1-\alpha\varphi^T\varphi})\mu^T\mu + 2(\frac{\gamma+1}{1-\alpha\varphi^T\varphi})\mu^Te + \eta e^Te\\& - \frac{1}{\alpha}[\gamma (2-\gamma )||\tilde{W}||_F^2-2\gamma (1-\gamma )W_M||\tilde{W}||_F-\gamma ^2W^2_M]\\& + 2\gamma  ||\varphi||W_M{\mu_M}.
		\end{split}
		\end{equation}
		
		Completing the squares for $\tilde{W}$, and considering inequality (\ref{eq:39}), and the fact that $||\varphi(x )|| \leq  \varphi_M$, one can have
		\begin{equation}\label{eq:40}
		\small
		\begin{split}
		\Delta V \leq & -(1-\alpha\varphi^T\varphi)|
		|\theta - \frac{1}{1-\alpha\varphi^T\varphi}G\tilde{\mathbf{x}}-\frac{\gamma-\alpha\varphi^T\varphi}{1-\alpha\varphi^T\varphi}\mu \\&-\frac{1}{1-\alpha\varphi^T\varphi}e
		||^2 -(1-\eta\bar{\sigma}^2(G))\bigg[||\tilde{\mathbf{x}}||^2
		\\&-\frac{1}{1-\eta\bar{\sigma}^2(G)} (2\rho_1||\tilde{\mathbf{x}}|| +\rho_2) \bigg]-\frac{1}{\alpha}\gamma (2-\gamma )
		\\&\bigg[||\tilde{W}||_F -\frac{1-\gamma }{2-\gamma }W_M  \bigg]^2,
		\end{split}
		\end{equation}
		where 
		\begin{equation}\label{eq:41}
		\begin{split}
		\rho_1 = & (\frac{\gamma +1}{1-\alpha\varphi_M^2})\bar{\sigma}(G)\mu_M+ \eta \bar{\sigma}(G)e_M,
		\end{split}
		\end{equation}
		where $e_M$ is derived using the result of Theorem 1 and is given as (\ref{eq:62}) in the Appendix. Moreover,
		\begin{equation}\label{eq:42}
		\small
		\begin{split}
		\rho_2 = & 2\gamma \mu_{M} W_M  +\frac{1}{\alpha}\frac{\gamma }{2-\gamma }W^2_M+ (-2\gamma  + \frac{(1+\gamma )^2}{1-\alpha\varphi_M^2})\mu^2_M\\& +2 \mu_M (\frac{\gamma+1}{1-\alpha\varphi_M^2})e_M + \eta e_M^2.
		\end{split}
		\end{equation}

		Then $\Delta V < 0$ as long as inequality (\ref{eq:31}), hold, and the quadratic term for $\tilde{x}$ in (\ref{eq:40}) is positive which is guaranteed when
		\begin{equation}\label{eq:43}
		\begin{split}
		||\tilde{\mathbf{x}}|| \geq & \frac{\rho_1+\sqrt{\rho^2_1+(1-\eta\bar{\sigma}^2(G))\rho_2}}{(1-\eta\bar{\sigma}^2(G))}  \triangleq \pi.
		\end{split}
		\end{equation}
		In general $\Delta V \leq 0$ in a compact set as long as (\ref{eq:31}) is satisfied and (\ref{eq:43}) holds. According to a standard Lyapunov extension theorem \cite{sarangapani2018neural}, this demonstrates that the observer error is UUB.
	\end{proof}

	%\vspace{0.2cm}
	%Theorem 2 shows that the formation error is UUB and  based on Lemma 1, the $\delta$ is bounded which means the formation is maintained and all the agents track the leader while marinating the formation. Moreover in this theorem we show that in attack-free case the observer error $\tilde{\mathbf{x}} $ is UUB by the bound $\sigma$. Now by considering observer error $\tilde{\mathbf{x}} $ as a residuae we can use this fact for detecting attack on each agent.

	Theorem 2 shows that when the multi-agent system reaches the desired formation, in the absence of attacks, the observer error is UUB, with bound $\pi$ (Theorem 2.4.6 in \cite{sarangapani2018neural}). That means in attack-free condition $||\tilde{x}|| < \pi$. By using this fact in attack detection, we consider the observer error as the attack residual with threshold $\pi$. 
	
	%By definition of UUB \cite{khalil2002nonlinear}, when $||\tilde{\mathbf{x}}|| > \pi$, the time derivative of $V$ is negative outside of the compact set $B_x = \big\{x~\big|~||x|| ~ < \pi \big\}$ or in other words, all the solutions that start outside of $B_x$ will enter this set within a finite time and will remain inside the set for all future times\cite{sarangapani2018neural}. That means in the attack-free scenario, $||\tilde{\mathbf{x}}|| \leq \pi$.

	\section{Attack detectability condition}\label{section:4}
	%The scalar $\pi$ is a threshold for the stacked state of a multi-agent system. To find the localized threshold by using norm inequality properties $||x||_{\infty} \leq ||x||_2$, in attack-free case we have
	%\begin{equation}\label{eq:44}
	%    ||\tilde{x}_i|| < \pi.
	%\end{equation}

	%The model of the system dynamics under attack can be written as:
	%\begin{equation} \label{eq:45}
	%    x_i ^{+} = f_i(x_i ) + u_i  + w_i  + s_i,
	%\end{equation} 
	%where, $s_i$ is the augmented effect of attack on $i$-th agent. From equations %(\ref{eq:15})-(\ref{eq:18}) and (\ref{eq:21}) the augmented attack effect $s_i $ is %given as
	%\begin{equation}\label{eq:46}
	%\begin{split}
	%   s_i  =  &\kappa_i u^a_i +  \hat{f}_i(x_i)-\hat{f}_i(x_i+\lambda_i x^a_i)+ c \lambda_i x^a_i + \\& c k_i \big[\sum_{j\in \mathcal{N}_i}a_{ij}(\lambda_i x^a_i 
	%  -\phi^i_j x^a_j ) + b_i(\lambda_i x^a_i )\big].
	%\end{split}
	%\end{equation}
	%\begin{remark}\label{remark:1}
	%Note that the state of agent $i$ is corrupted by the attack, but it is still different from the $x^c_i $ in (\ref{eq:16}). The agent's state cannot employ by its controller or broadcast to its neighbours. Thus, the sensor data which is under attack is employed by the controller and broadcasted to neighbours.
	%\end{remark}
	
	%By using the effect of attack which mentioned in (\ref{eq:20}), the system model under attack can be defined as (55). 

	The scalar $\pi$ is a threshold for the stacked observer error of a multi-agent system. Using the norm inequality property $||x||_{\infty} \leq ||x||_2$, the scalar $\pi$  given by (\ref{eq:43}) is each agent's observer error's threshold. The residual dynamics under attack can be written as:
	\begin{equation} \label{eq:45}
	\Tilde{x}_i ^{+} = G_i\Tilde{x}_i + \Tilde{W}_i^{T} \varphi(x_i ) + \epsilon_i  + w_i + e_i + s_i,
	\end{equation} 
	where, $s_i$ is the overall attacks effect  on residual of $i$-th agent. From equations (\ref{eq:15})-(\ref{eq:18}), (\ref{eq:20}) and (\ref{eq:21}) the overall attacks effect $s_i $ is given as
	\begin{equation}\label{eq:46}
	\begin{split}
	s_i  =  &\kappa_i u^a_i + \lambda_i G_i x^a_i+  \hat{f}_i(x_i)-\hat{f}_i(x_i+\lambda_i x^a_i) + \\& \sum_{j\in \mathcal{N}_i}a_{ij}(\lambda_i x^a_i 
	-\phi^i_j \bar{x}                ^a_j ) + b_i(\lambda_i x^a_i ).
	\end{split}
	\end{equation}
	
	The response  of  residual  signal (\ref{eq:25}) under attack when $x_i(0) = \hat{x}_i(0)$ can be presented as
	\begin{equation}\label{eq:47}
	\tilde{x}_i = \sum^{k-1}_{l=0}G_i^{k-l-1}(\Tilde{W}_i^{T} \varphi(x_i ) + e_i + \epsilon_i  + w_i  + s_i ).
	\end{equation}
	The following Theorem determines the detectability condition.
	% Now, by finding the effect of the attack in the system dynamics, which is mentioned in equations (\ref{eq:46}), the following detectability condition is determined.
	
	\begin{theorem}
		Consider the $i$-th agent system model described by (\ref{eq:2}) and the NN observer (\ref{eq:20}). The agent can detect the attacks if the overall attacks effect  $s_i$ satisfies
		\begin{equation}\label{eq:48}
		\begin{split}
		||\sum^{k-1}_{l=0}&G_i^{k-l-1}s_i || \geq \pi + \\&
		||\sum^{k-1}_{l=0}G_i^{k-l-1}(\Tilde{W}_i^{T} \varphi(x_i ) + e_i + \epsilon_i  + w_i )||.
		\end{split}
		\end{equation}
		
	\end{theorem}
	\vspace{0.2cm}
	\begin{proof}
		Similar proof can be found in \cite{niu2019attack}.
	\end{proof}
	\color{black}
	
	\begin{remark}\label{remark:2}
		The attack effect $s_i$ can be due to any type of attacks mentioned earlier including some combination of them.
	\end{remark}
	
	\begin{remark}
		Note that if $d_i =0$ for $i=1,...,N$ in equation (\ref{eq:6}), the formation control setup becomes a consensus problem \cite{oh2015survey}, where the proposed attack detection method can still be applied.
	\end{remark} 
	%\begin{equation}\label{eq:47}
	%    \tilde{x}_i = \sum^{l=k-1}_{l=0}g_i^{k-l-1}(\Tilde{W}_i^{T} \varphi(x_i ) + e_i + \epsilon_i  + w_i  + s_i ).
	%\end{equation}

	%As a result, the attack is detectable when we have the above condition for the effect of attack $s_i $.
	
	\begin{figure*}[]\label{fig:2}
		\centering % <-- added
		\begin{subfigure}{0.32\textwidth}
			\includegraphics[width=\linewidth,height=3.9cm]{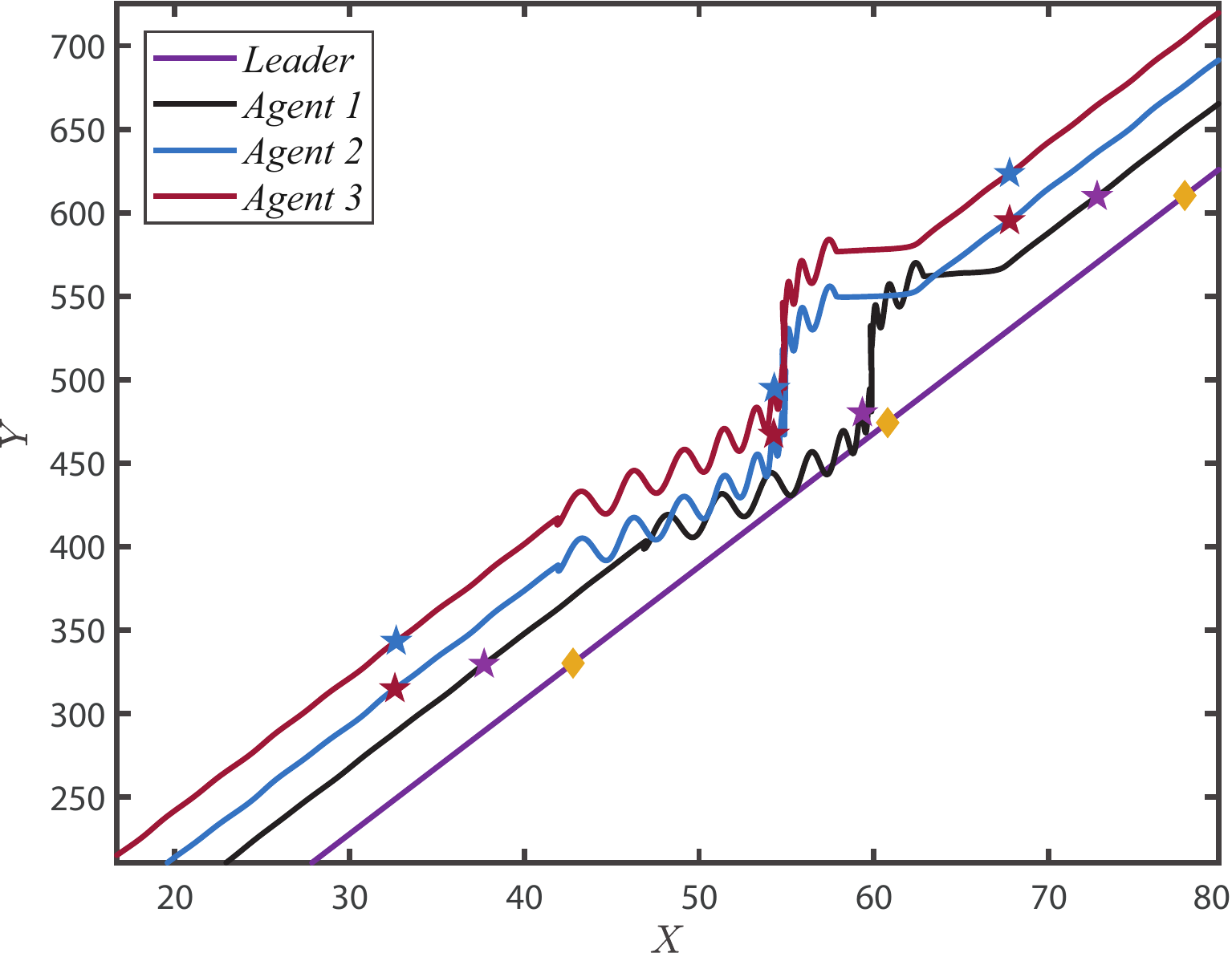}
			\caption{}
			\label{fig:8a}
		\end{subfigure}\hfil % <-- added
		\begin{subfigure}{0.32\textwidth}
			\includegraphics[width=\linewidth,height=3.9cm]{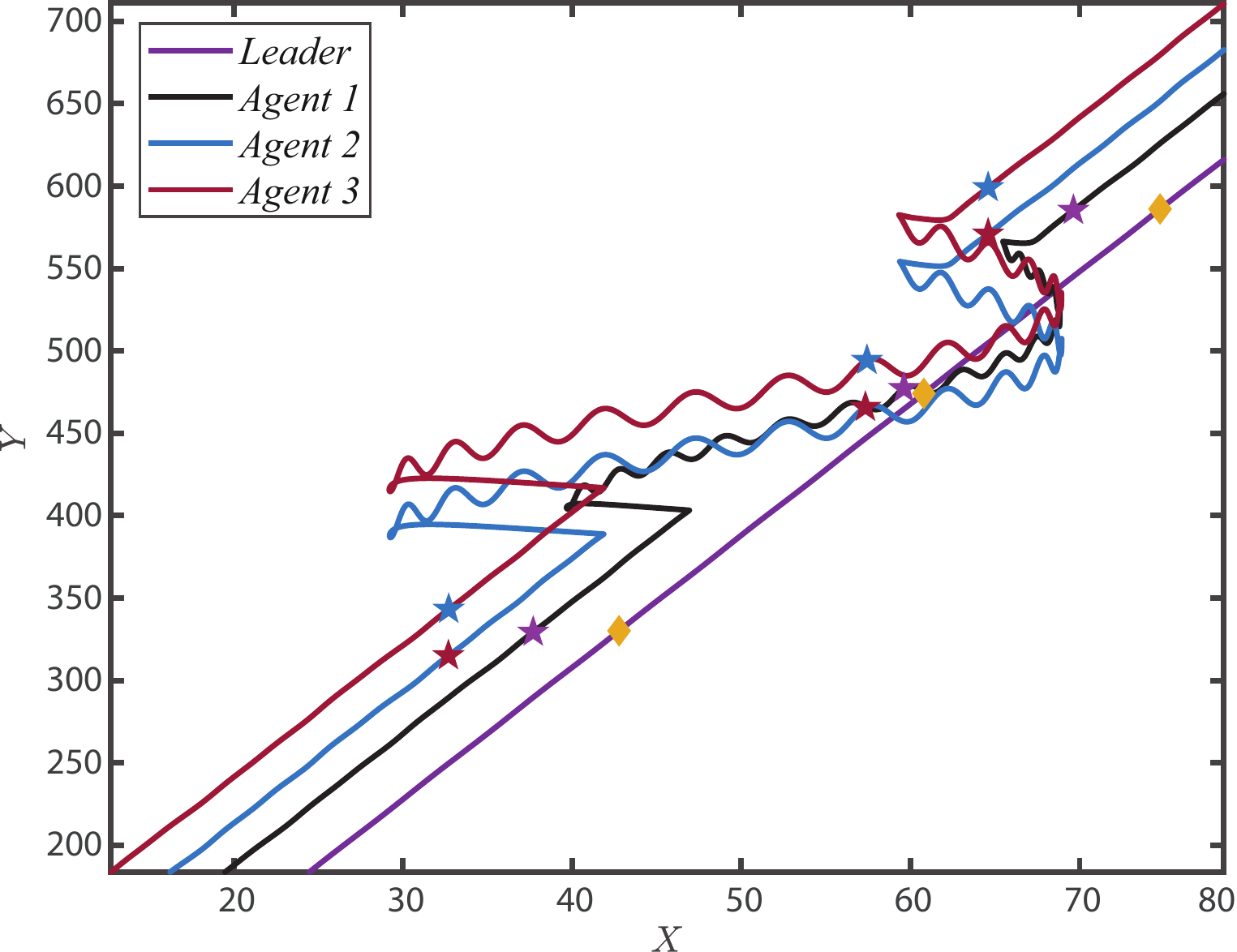}
			\caption{}
			\label{fig:8b}
		\end{subfigure}\hfil % <-- added
		\begin{subfigure}{0.32\textwidth}
			\includegraphics[width=\linewidth,height=3.9cm]{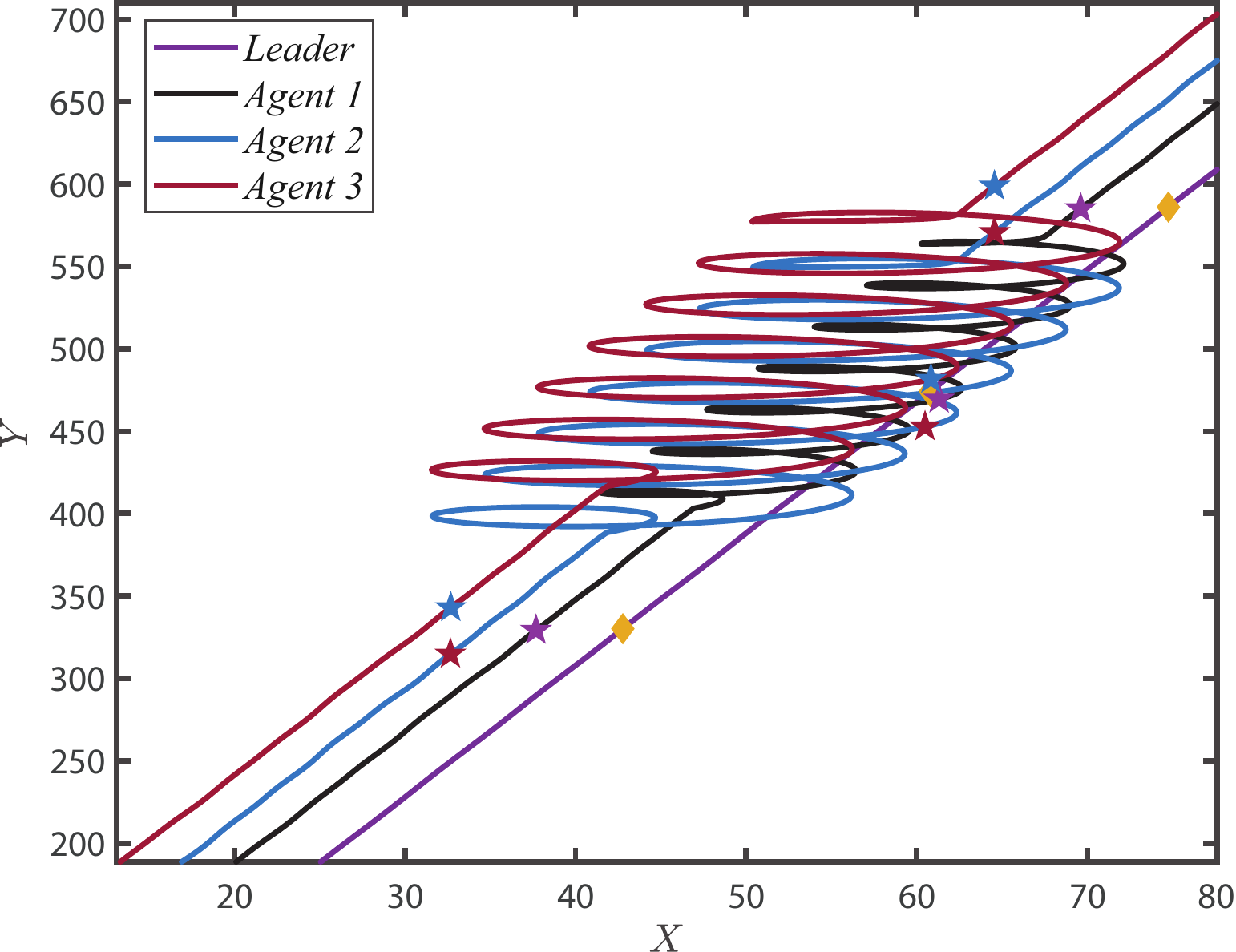}
			\caption{}
			\label{fig:8c}
		\end{subfigure}
		
		\caption{System trajectory under attack. (a) Attack on actuator channel of agent 1. (b) Attack on sensor channel of agent 3. (c) Attack on the neighbouring channel between agent 2 and agent 3.}
		\label{fig:8}
	\end{figure*}
	
	\section{Simulation Results}\label{section:5}
	
	We conducted simulations to evaluate the proposed detection method's performance and to show that the proposed detection method meets the stated objective, thus enabling each agent to detect attacks on its communication channels. 
	
	\textit{Example 1}: Consider a multi-agent system with the leader and three agents as followers, modelled by directed graph Fig. 3.
	\begin{equation}\label{eq:49}
	x_i ^{+} = \begin{bmatrix}
	\frac{x_{i2}}{1+x^2_{i2}}\\
	\frac{x_{i1}}{1+x^2_{i2}}
	\end{bmatrix} + u_i + w_i,
	\end{equation}
	where $x_i = [x_{i1}~x_{i2}]^T$. The variable $x_{i1}$ and $x_{i2}$ are the position in $x$ and $y$ dimension, respectively. Select the control parameter $k_i = 0.2I_2$, $c = 0.7$, and the observer gain as $G_i = 0.23I_2$. The sampling period is chosen as $T = 1~ms$. The initial conditions for the follower agents are: $x_1(0)=(1, -1)^T$, $x_2(0)=(3,4)^T$, and $x_3(0)=(3,-5)^T$. The disturbance for each agent is considered as follows
	\begin{equation}\label{eq:50}
	\begin{split}
	& w_1= [0.01,0.05]^T sin(2t),\\
	& w_2= [0.02,0.05]^T cos(3t),\\
	& w_3= [0.02,0.01]^T sin(3t).
	\end{split}
	\end{equation}
	
	\begin{figure}[tb]\label{fig:3}
		\centering
		\includegraphics[scale=0.50]{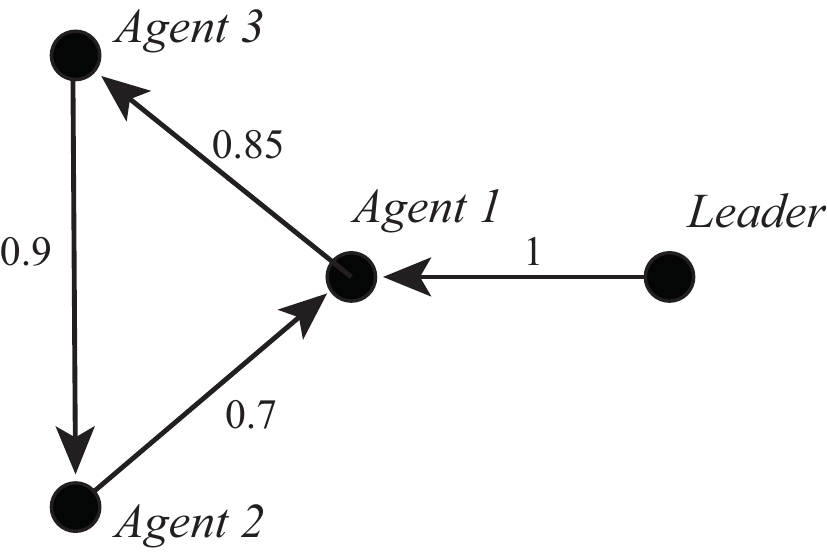}
		\caption{Communication topology and formation shape.}
		\label{fig:universe1}
	\end{figure} 
	
	RBFNN is selected with 9 neurons with centers $m_j$ evenly spaced for each agent and $F_i = 0.1I_{9}$, $\alpha = 0.1$. 
	In the simulation, we test our proposed method for three types of attacks to show the detection system performance. We consider three scenarios of injecting false data to actuator channel, sensor channel, and neighbouring channel. The detection system shows that it is capable of detecting these attacks separately; therefore, it can detect the combination of the mentioned attacks. The attack threshold (\ref{eq:43}) is $\pi = 0.44$, for the attack-free case. The desired formation and communication topology is shown in Fig. 3, and the desired position of each agent with respect to the leader is given as follows:
	\begin{equation}\label{eq:51}
	\begin{split}
	& d_1: [5,~0]^T, d_2: [10,~14]^T,  d_3: [-10,~14]^T.
	\end{split}
	\end{equation}

	\begin{figure}[t]\label{fig:4}
		\centering
		\includegraphics[scale=0.45]{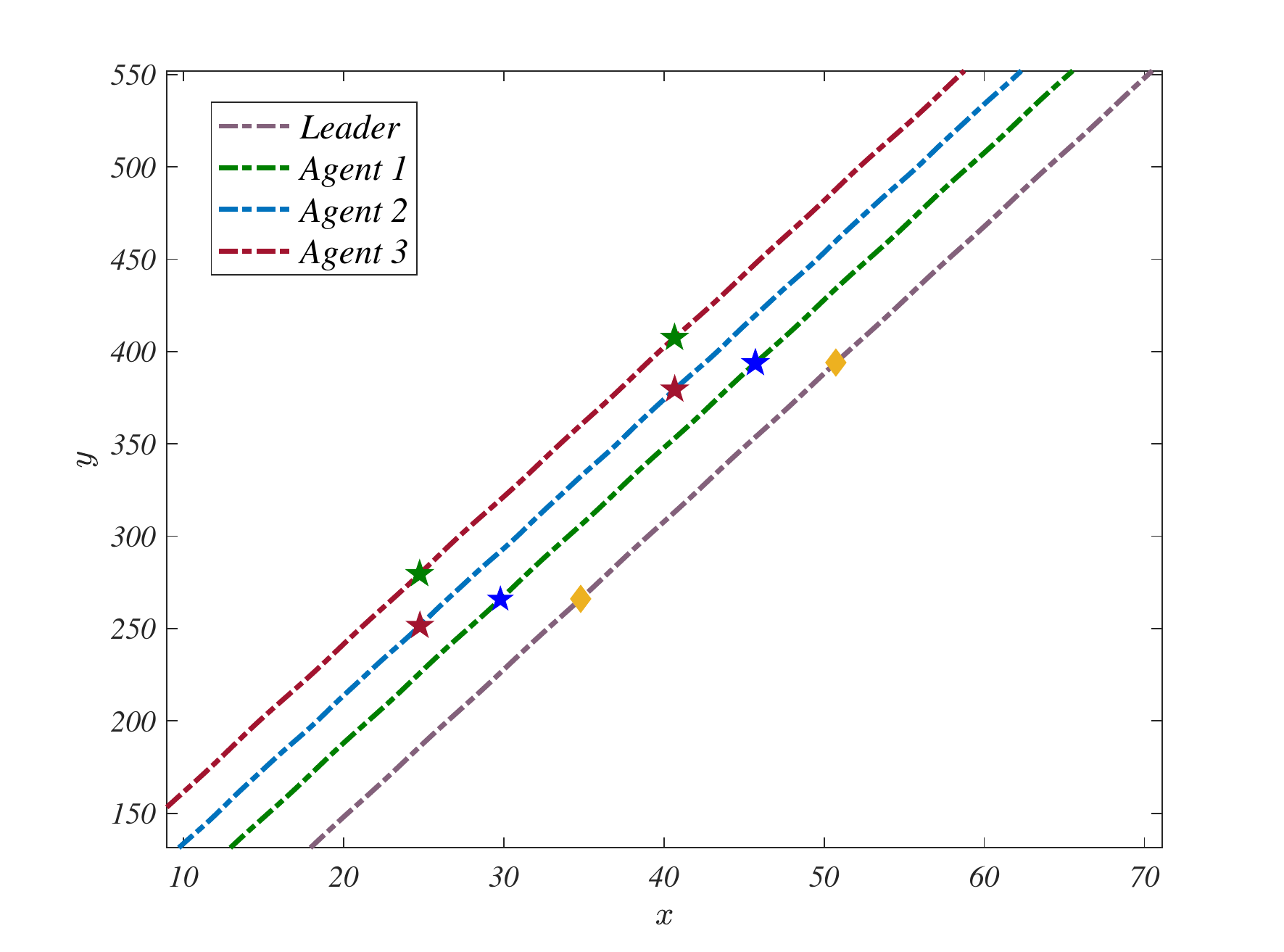}
		\caption{System trajectory in attack-free case.}
		\label{fig:universe2}
	\end{figure}  
	\vspace{0.2cm}
	\noindent \textbf{Case 1: Attack on actuator channel}

	For the first case, we consider the following leader trajectory 
	\begin{equation}\label{eq:52}
	\begin{split}
	& x_{l1} = t+2,\\
	& x_{l2} = 8t+4,
	\end{split}
	\end{equation}
	and we assume the attacker injects a false data on actuator channel of agent 1:
	\begin{equation}\label{eq:53}
	u^c_{1}  = u_1 + \kappa_3  u^a_1,
	\end{equation}
	where $\kappa_1  = 1$ when $50 \leq t\leq 70$ and $u^a_1  = [2sin(t/4),3sin(4t)]^T$.

	In Fig. 4, the system formation is shown for the attack-free case, and in Fig. 2(a), the effect of attack can be seen on the system formation. The residual signal of agent 1, which is $\tilde{x}_1$ with the threshold $\pi$ is shown in Fig. 5, which shows the attack increases the residual signal of agent 1 and it exceeds the threshold.
	%, since the agent 1 sends its data to agent 2, this agent  is also exposed to attack
	
	\vspace{0.2cm}
	
	\noindent \textbf{Case 2: Attack on sensor channel}
	
	We consider that trajectory (\ref{eq:52}) for the leader, and we assume the attacker injects false data on the sensor channel of agent 3:
	\begin{equation}\label{eq:54}
	x^c_{3}  = x_3  + \lambda_3  x^a_3,
	\end{equation}
	where $\lambda_3  = 1$ when $50 \leq t \leq 70$ and $x^a_3  = [4sin(t/4),5sin(5t)]^T$.

	In Fig. 2(b), the effect of the attack on agent 1 sensor channel can be seen. The residual signal of agent 3 is $\tilde{x}_3$ increases in the attack duration and reveals the attack (Fig. 6).

	\vspace{0.2cm}
	\noindent \textbf{Case 3: Attack on neighbouring channel}
	
	We consider that the leader trajectory in (\ref{eq:52}), and we assume the attacker injects false data on the neighbouring channel of agent 2. We consider that the attacker injects the false data into the communication channel between agent 3 and agent 2, which means the attacker changes the data that agent 2 receives from agent 3:
	
	\begin{equation}\label{eq:55}
	\bar{x}^c_{3}  = \bar{x}_3  + \phi^2_3  \bar{x}^a_3,
	\end{equation}
	where $\phi^2_3  = 1$ when $50 \leq t \leq 70$ and $\bar{x}^a_3  = [-4sin(t), 3cos(t)]^T$.

	Fig. 2(c) shows the effect of the attack on the neighbouring channel of agent 2 is shown, and the residual signal of agent 2 is shown in Fig. 7, which shows the detection of the attack.

	\begin{figure}[t]\label{fig:5}
		\centering
		\includegraphics[scale=0.45]{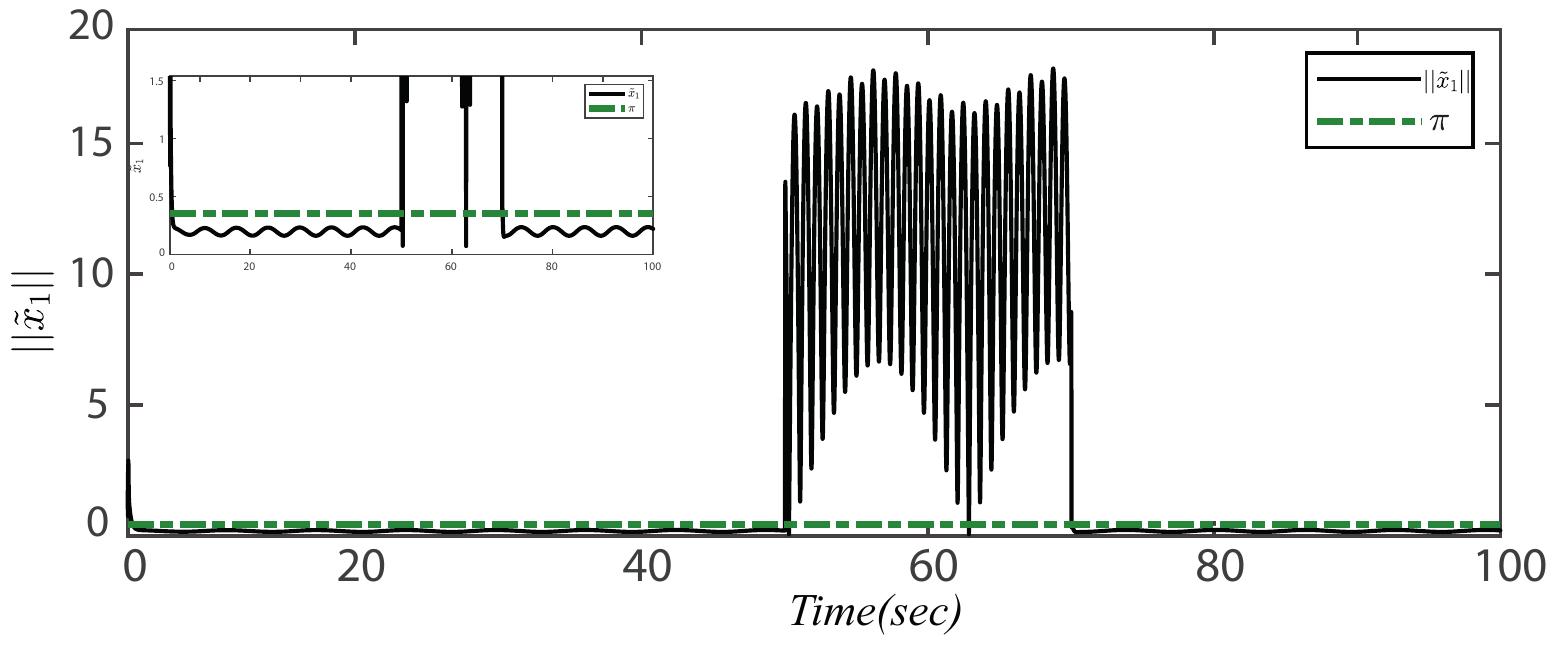}
		\caption{Residual signal under attack on actuator channel.}
		\label{fig:universe3}
	\end{figure}

	\begin{figure}[t]\label{fig:6}
		\centering
		\includegraphics[scale=0.5]{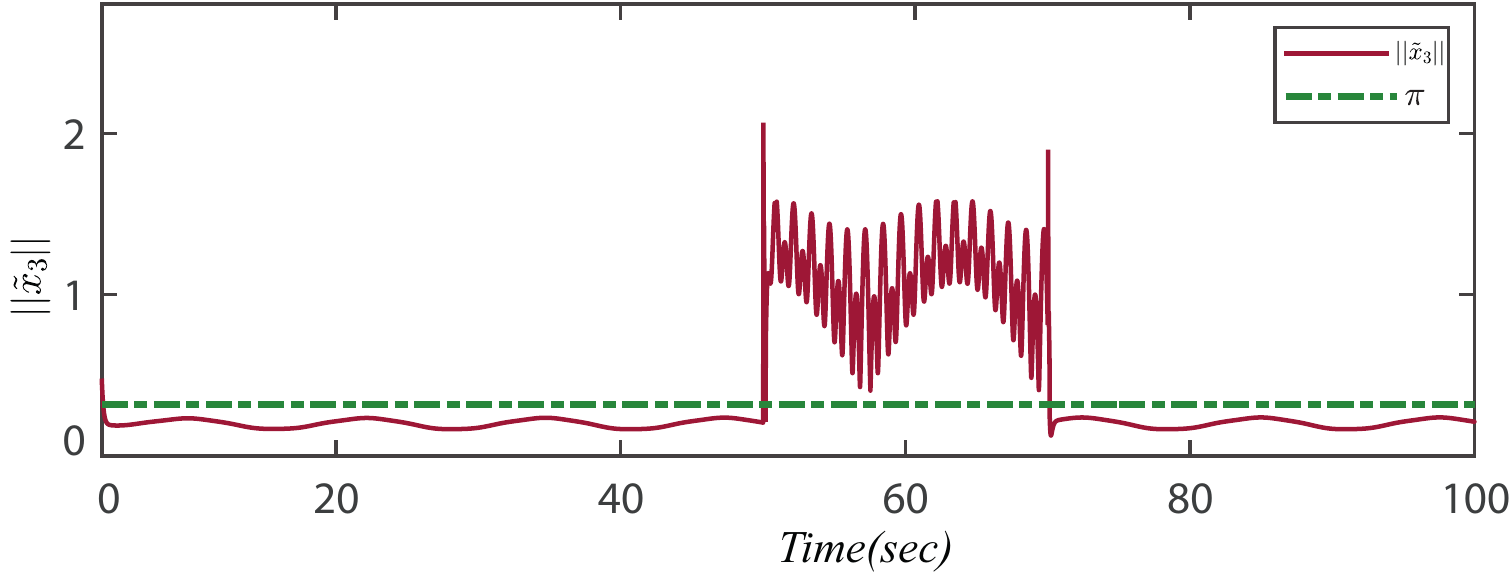}
		\caption{Residual signal under attack on sensor channel.}
		\label{fig:universe4}
	\end{figure}

	\begin{figure}[t]\label{fig:7}
		\centering
		\includegraphics[scale=0.5]{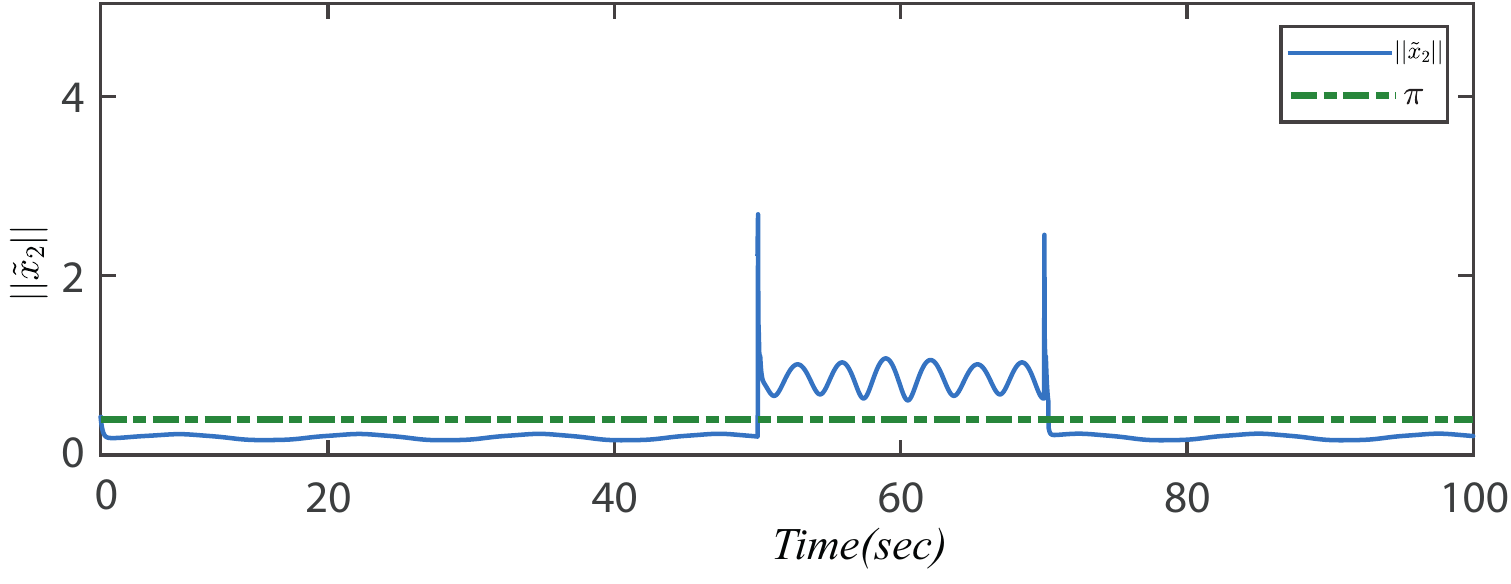}
		\caption{Residual signal under attack on neighbouring channel.}
		\label{fig:universe5}
	\end{figure}

	%\begin{figure}[tb]
	%   \centering
	%    \includegraphics[scale=0.45]{image/ftilde.eps}
	%   \caption{NN estimation error for agent 2.}
	%  \label{fig:universe6}
	%\end{figure}  

	%Moreover, in Fig. 8 the $\tilde{f}_2(x_2 )$ is shown and it shows the functionality of NN. As shown in this figure the NN error decreases which shows that the NN can estimate the function ${f}_2(x_2 )$ properly.
	
	\section{Conclusion}
	We developed a cyber-attack detection system for a class of discrete, nonlinear, heterogeneous, multi-agent systems in a formation control setting. We proposed a formation control for a class of nonlinear discrete multi-agent systems. We investigated the false data injection attack into agents' communication channels, where a NN-based attack detection on the actuator, sensor, and neighbouring communication channel was developed. The Lyapunov stability analysis was used to prove UUB of formation error, NN weigh matrix error, and observer error. The simulation results were presented and included examples of attacks on communication between neighbors, sensors, and actuators, showing the functionality and performance of the attack detection system and NN-based observer.

	\section{Appendix}
	\textit{Proof of Theorem 1:}
	Let us consider the following Lyapunov candidate function
	
	\begin{equation}\label{eq:56}
	V = \frac{1}{\bar{\sigma}(\bar{L}^T\bar{L})}e^Te + \frac{1}{\alpha}tr(\tilde{W}^T\tilde{W)}.
	\end{equation}
	
	Let us define $\nu = \underline{1} \otimes f_l(x_l) + d$. Based on the Assumption 2, the $\nu$ is bounded by 
	\begin{equation}
	||\nu|| \leq \nu_M,
	\end{equation}
	with $\nu_M = F_M + d_M$.

	Therefore, from equations (\ref{eq:9-1}) and (\ref{eq:26}) the  first  difference of  Lyapunov  candidate  is  defined  as follows
	
	\begin{equation}\label{eq:57}
	\small
	\begin{split}
	\Delta V = & -(2 - \alpha\varphi^T\varphi)\theta^T\theta - 2(1 - \alpha\varphi^T\varphi)\theta^T\mu  +\alpha\varphi^T\varphi\mu^T\mu\\& - \frac{1}{\alpha}[\gamma (2-\gamma )||\tilde{W}||_F^2-2\gamma (1-\gamma )W_M||\tilde{W}||_F-\gamma ^2W^2_M]\\& + 2\gamma  ||\varphi||W_M{\mu_M}
	- \frac{1}{\bar{\sigma}(\bar{L}^T\bar{L})}(1- c^2 P^TP)e^Te +\\& \frac{1}{\bar{\sigma}(\bar{L}^T\bar{L})}\theta^T\bar{L}^T\bar{L}\theta + \frac{1}{\bar{\sigma}(\bar{L}^T\bar{L})}(\mu + \nu)^T\bar{L}^T\bar{L}(\mu+\nu) \\& - \frac{2}{\bar{\sigma}(\bar{L}^T\bar{L})} c\theta^T\bar{L}^TPe
	- \frac{2}{\bar{\sigma}(\bar{L}^T\bar{L})} c (\mu^T+\nu^T)\bar{L}^TPe \\&+ \frac{2}{\bar{\sigma}(\bar{L}^T\bar{L})}(\mu^T+\nu^T)\bar{L}^T\bar{L}\theta,
	\end{split}
	\end{equation}
	recognizing the terms in (\ref{eq:57}) yields
	\begin{equation}\label{eq:58}
	\small
	\begin{split}
	\Delta V \leq & -(1 - \alpha\varphi^T\varphi)\theta^T\theta - 2(1 - \alpha\varphi^T\varphi)\theta^T\mu \\& +\alpha\varphi^T\varphi\mu^T\mu + 2\gamma\mu\theta - \frac{1}{\alpha}[\gamma (2-\gamma )||\tilde{W}||_F^2-\\&2\gamma (1-\gamma )W_M||\tilde{W}||_F-\gamma ^2W^2_M] + 2\gamma  ||\varphi||W_M{\mu_M}\\&
	- \frac{1}{\bar{\sigma}(\bar{L}^T\bar{L})}(1- c^2 P^TP)e^Te + (\mu + \nu)^T(\mu+\nu)\\& - \frac{2}{\bar{\sigma}(\bar{L}^T\bar{L})} c \theta^T\bar{L}^TPe - \frac{2}{\bar{\sigma}(\bar{L}^T\bar{L})} c(\mu^T+\nu^T)\bar{L}^TPe \\& + \frac{2}{\bar{\sigma}(\bar{L}^T\bar{L})}(\mu^T+\nu^T)\bar{L}^T\bar{L}\theta.
	\end{split}
	\end{equation}

	%\begin{equation}
	%\small
	%\begin{split}
	%\Delta V \leq & -(1 - \alpha\varphi^T\varphi)||\theta + (1 %-\frac{1+\gamma}{1 - \alpha\varphi^T\varphi})\mu + 
	%\\& 
	% \frac{1}{\bar{\sigma}(\bar{L}^T\bar{L})(1 - \alpha\varphi^T\varphi)}\bar{L}^TPe||^2+\alpha\varphi^T\varphi\mu^T\mu +
	%  \\& 
	%   2\gamma||\varphi||W_M{\mu_M}
	%    - \frac{1}{\bar{\sigma}(\bar{L}^T\bar{L})}(1- P^TP)||e||^2 +\mu^T\mu + \nu^T\nu
	%%  \\&
	% - \frac{2}{\bar{\sigma}(\bar{L}^T\bar{L})}(\mu^T+\nu^T)\overline{L}^TPe + 2\mu^T\nu
	%\\& -\frac{1}{\alpha}\gamma (2-\gamma ) \bigg[||\tilde{W}|| %-\frac{1-\gamma }{2-\gamma }W_M  \bigg]^2,
	%\end{split}
	%\end{equation}

	Completing the squares for $\theta$ and $\tilde{W}$, and considering   (\ref{eq:58}) one can have:

	\begin{equation}\label{eq:59}
	\small
	\begin{split}
	\Delta V \leq & -(1 - \alpha\varphi^T\varphi)||\theta  - \frac{\gamma +\alpha\varphi^T\varphi }{1-\alpha\varphi^T\varphi}\mu - \frac{1}{1-\alpha\varphi^T\varphi}\nu+
	\\& 
	\frac{1}{\bar{\sigma}(\bar{L}^T\bar{L})(1 - \alpha\varphi^T\varphi)} c\bar{L}^TPe||^2
	- \frac{1}{\bar{\sigma}(\bar{L}^T\bar{L})}(1-\\& \eta  c^2\bar{\sigma}(P^T P)) [||e||^2 -\frac{\bar{\sigma}(\bar{L}^T\bar{L})}{1-\eta  c^2\bar{\sigma}(P^T P)}( 2\Lambda_1||e|| + \Lambda_2)]  \\&-\frac{1}{\alpha}\gamma (2-\gamma ) \bigg[||\tilde{W}||_F -\frac{1-\gamma }{2-\gamma }W_M  \bigg]^2,
	\end{split}
	\end{equation}
	where $\Lambda_1$ and $\Lambda_2$ are  given as follows: 
	
	\begin{equation}\label{eq:60}
	\small
	\begin{split}
	\Lambda_1 = & \frac{1}{\bar{\sigma}(\bar{L}^T\bar{L})} c\bar{\sigma}(P)\bar{\sigma}(\bar{L}^T)\big[(\frac{\gamma +1}{1-\alpha\varphi_M^2})\mu_M + \frac{2-\alpha}{1-\alpha\varphi_M^2}\nu_M\big],
	\end{split}
	\end{equation}
	and,
	\begin{equation}\label{eq:61}
	\begin{split}
	\Lambda_2 =& 2\gamma W_M\mu_M+\frac{1}{\alpha}\frac{\gamma }{2-\gamma }W^2_M+ (-2\gamma+\\& \frac{(1+\gamma )^2}{1-\alpha\varphi_M^2})\mu^2_M + \frac{2(1+\gamma)}{1-\alpha\varphi_M^2}\mu_m\nu_M + \frac{2-\alpha}{1-\alpha\varphi_M^2}\nu^2_M.
	\end{split}
	\end{equation}
	
	Then $\Delta V < 0$ as long as (\ref{eq:28})-(\ref{eq:30}) hold, and the quadratic term for $e$ in (\ref{eq:59}) is positive which is guaranteed when
	
	\begin{equation}\label{eq:62}
	\begin{split}
	||e|| \geq & \frac{\Lambda_1+\sqrt{\Lambda^2_1+(1-\eta c^2\bar{\sigma}(P^T P))\Lambda_2}}{(1-\eta c^2\bar{\sigma}(P^T P))}\triangleq e_M.
	\end{split}
	\end{equation}
	By using (\ref{eq:58}) and completing the squares for $e$ we can conclude that $\Delta V < 0$ as long as

	\begin{equation}\label{eq:63}
	\small
	||\tilde{W}||_F \geq \frac{\gamma(1-\gamma )W_M+\sqrt{\gamma^2(1-\gamma )^2W^2_M+\gamma(2-\gamma )\xi}}{\gamma(2-\gamma)},
	\end{equation}
	where $\xi$ is

	\begin{equation}\label{eq:64}
	\small
	\begin{split}
	\xi = & {2\alpha\gamma W_M\mu_M} + \gamma^2 W^2_M+(-2\alpha\gamma+ \frac{\alpha(1+\gamma )^2}{1-\alpha\varphi_M^2})\mu^2_M +\\& \frac{2\alpha(1+\gamma)}{1-\alpha\varphi_M^2}\mu_m\nu_M + \frac{\alpha(2-\alpha)}{1-\alpha\varphi_M^2}\nu^2_M +\frac{\bar{\sigma}(\bar{L}^T\bar{L})}{1-\eta c^2\bar{\sigma}(P^T P)}\Lambda^2_1.
	\end{split}
	\end{equation}
	
	In general $\Delta V \leq 0$ in a compact set as long as (\ref{eq:28}) through (\ref{eq:30}) are satisfied and either (\ref{eq:62}) or (\ref{eq:63}) holds. According to a standard Lyapunov extension theorem \cite{sarangapani2018neural}, this demonstrates that the formation error and weight estimates error are UUB. Lemma 1 shows that the tracking error vector $\delta(t)$ is UUB. That means all agents follow the leader while maintaining the desired formation.
	
	%%%%%%%%%%%%%%%%%%%%%%%%%%%%%%%%%%%%%%%%%%%%%%%%%%%%%%%%%%%%%%%%%%%%%%%%%%%%%%%%

	%%%%%%%%%%%%%%%%%%%%%%%%%%%%%%%%%%%%%%%%%%%%%%%%%%%%%%%%%%%%%%%%%%%%%%%%%%%%%%%%

	%%%%%%%%%%%%%%%%%%%%%%%%%%%%%%%%%%%%%%%%%%%%%%%%%%%%%%%%%%%%%%%%%%%%%%%%%%%%%%%%

	\bibliographystyle{IEEEtran}
	\bibliography{Journal.bib}

\begin{thebibliography}{10}
\providecommand{\url}[1]{#1}
\csname url@samestyle\endcsname
\providecommand{\newblock}{\relax}
\providecommand{\bibinfo}[2]{#2}
\providecommand{\BIBentrySTDinterwordspacing}{\spaceskip=0pt\relax}
\providecommand{\BIBentryALTinterwordstretchfactor}{4}
\providecommand{\BIBentryALTinterwordspacing}{\spaceskip=\fontdimen2\font plus
\BIBentryALTinterwordstretchfactor\fontdimen3\font minus
  \fontdimen4\font\relax}
\providecommand{\BIBforeignlanguage}[2]{{%
\expandafter\ifx\csname l@#1\endcsname\relax
\typeout{** WARNING: IEEEtran.bst: No hyphenation pattern has been}%
\typeout{** loaded for the language `#1'. Using the pattern for}%
\typeout{** the default language instead.}%
\else
\language=\csname l@#1\endcsname
\fi
#2}}
\providecommand{\BIBdecl}{\relax}
\BIBdecl

\bibitem{lee2008cyber}
E.~A. Lee, ``Cyber physical systems: Design challenges,'' in \emph{2008 11th
  IEEE International Symposium on Object and Component-Oriented Real-Time
  Distributed Computing (ISORC)}.\hskip 1em plus 0.5em minus 0.4em\relax IEEE,
  2008, pp. 363--369.

\bibitem{xiong2015cyber}
G.~Xiong, F.~Zhu, X.~Liu, X.~Dong, W.~Huang, S.~Chen, and K.~Zhao,
  ``Cyber-physical-social system in intelligent transportation,''
  \emph{IEEE/CAA Journal of Automatica Sinica}, vol.~2, no.~3, pp. 320--333,
  2015.

\bibitem{chen2017cyber}
B.~Chen, Z.~Yang, S.~Huang, X.~Du, Z.~Cui, J.~Bhimani, X.~Xie, and N.~Mi,
  ``Cyber-physical system enabled nearby traffic flow modelling for autonomous
  vehicles,'' in \emph{2017 IEEE 36th International Performance Computing and
  Communications Conference (IPCCC)}.\hskip 1em plus 0.5em minus 0.4em\relax
  IEEE, 2017, pp. 1--6.

\bibitem{macana2011survey}
C.~A. Macana, N.~Quijano, and E.~Mojica-Nava, ``A survey on cyber physical
  energy systems and their applications on smart grids,'' in \emph{2011 IEEE
  PES conference on innovative smart grid technologies Latin America (ISGT
  LA)}.\hskip 1em plus 0.5em minus 0.4em\relax IEEE, 2011, pp. 1--7.

\bibitem{pasqualetti2013attack}
F.~Pasqualetti, F.~D{\"o}rfler, and F.~Bullo, ``Attack detection and
  identification in cyber-physical systems,'' \emph{IEEE Transactions on
  Automatic Control}, vol.~58, no.~11, pp. 2715--2729, 2013.

\bibitem{baheti2011cyber}
R.~Baheti and H.~Gill, ``Cyber-physical systems,'' \emph{The impact of control
  technology}, vol.~12, no.~1, pp. 161--166, 2011.

\bibitem{kriaa2012modeling}
S.~Kriaa, M.~Bouissou, and L.~Pi{\`e}tre-Cambac{\'e}d{\`e}s, ``Modeling the
  stuxnet attack with bdmp: Towards more formal risk assessments,'' in
  \emph{2012 7th International Conference on Risks and Security of Internet and
  Systems (CRiSIS)}.\hskip 1em plus 0.5em minus 0.4em\relax IEEE, 2012, pp.
  1--8.

\bibitem{case2016analysis}
D.~U. Case, ``Analysis of the cyber attack on the ukrainian power grid,''
  \emph{Electricity Information Sharing and Analysis Center (E-ISAC)}, 2016.

\bibitem{miao2016coding}
F.~Miao, Q.~Zhu, M.~Pajic, and G.~J. Pappas, ``Coding schemes for securing
  cyber-physical systems against stealthy data injection attacks,'' \emph{IEEE
  Transactions on Control of Network Systems}, vol.~4, no.~1, pp. 106--117,
  2016.

\bibitem{boem2017distributed}
F.~Boem, A.~J. Gallo, G.~Ferrari-Trecate, and T.~Parisini, ``A distributed
  attack detection method for multi-agent systems governed by consensus-based
  control,'' in \emph{2017 IEEE 56th Annual Conference on Decision and Control
  (CDC)}.\hskip 1em plus 0.5em minus 0.4em\relax IEEE, 2017, pp. 5961--5966.

\bibitem{farivar2019artificial}
F.~Farivar, M.~S. Haghighi, A.~Jolfaei, and M.~Alazab, ``Artificial
  intelligence for detection, estimation, and compensation of malicious attacks
  in nonlinear cyber-physical systems and industrial iot,'' \emph{IEEE
  Transactions on Industrial Informatics}, vol.~16, no.~4, pp. 2716--2725,
  2019.

\bibitem{jin2017adaptive}
X.~Jin, W.~M. Haddad, and T.~Yucelen, ``An adaptive control architecture for
  mitigating sensor and actuator attacks in cyber-physical systems,''
  \emph{IEEE Transactions on Automatic Control}, vol.~62, no.~11, pp.
  6058--6064, 2017.

\bibitem{ding2016observer}
D.~Ding, Z.~Wang, D.~W. Ho, and G.~Wei, ``Observer-based event-triggering
  consensus control for multiagent systems with lossy sensors and
  cyber-attacks,'' \emph{IEEE Transactions on Cybernetics}, vol.~47, no.~8, pp.
  1936--1947, 2016.

\bibitem{olfati2007consensus}
R.~Olfati-Saber, J.~A. Fax, and R.~M. Murray, ``Consensus and cooperation in
  networked multi-agent systems,'' \emph{Proceedings of the IEEE}, vol.~95,
  no.~1, pp. 215--233, 2007.

\bibitem{xiao2009finite}
F.~Xiao, L.~Wang, J.~Chen, and Y.~Gao, ``Finite-time formation control for
  multi-agent systems,'' \emph{Automatica}, vol.~45, no.~11, pp. 2605--2611,
  2009.

\bibitem{oh2015survey}
K.-K. Oh, M.-C. Park, and H.-S. Ahn, ``A survey of multi-agent formation
  control,'' \emph{Automatica}, vol.~53, pp. 424--440, 2015.

\bibitem{teixeira2010networked}
A.~Teixeira, H.~Sandberg, and K.~H. Johansson, ``Networked control systems
  under cyber attacks with applications to power networks,'' in
  \emph{Proceedings of the 2010 American Control Conference}.\hskip 1em plus
  0.5em minus 0.4em\relax IEEE, 2010, pp. 3690--3696.

\bibitem{khazraei2017replay}
A.~Khazraei, H.~Kebriaei, and F.~R. Salmasi, ``Replay attack detection in a
  multi agent system using stability analysis and loss effective
  watermarking,'' in \emph{2017 American Control Conference (ACC)}.\hskip 1em
  plus 0.5em minus 0.4em\relax IEEE, 2017, pp. 4778--4783.

\bibitem{pang2016two}
Z.-H. Pang, G.-P. Liu, D.~Zhou, F.~Hou, and D.~Sun, ``Two-channel false data
  injection attacks against output tracking control of networked systems,''
  \emph{IEEE Transactions on Industrial Electronics}, vol.~63, no.~5, pp.
  3242--3251, 2016.

\bibitem{barboni2020detection}
A.~Barboni, H.~Rezaee, F.~Boem, and T.~Parisini, ``Detection of covert
  cyber-attacks in interconnected systems: A distributed model-based
  approach,'' \emph{IEEE Transactions on Automatic Control}, 2020.

\bibitem{arrichiello2015observer}
F.~Arrichiello, A.~Marino, and F.~Pierri, ``Observer-based decentralized fault
  detection and isolation strategy for networked multirobot systems,''
  \emph{IEEE Transactions on Control Systems Technology}, vol.~23, no.~4, pp.
  1465--1476, 2015.

\bibitem{huang2019reliable}
X.~Huang and J.~Dong, ``Reliable leader-to-follower formation control of
  multiagent systems under communication quantization and attacks,'' \emph{IEEE
  Transactions on Systems, Man, and Cybernetics: Systems}, vol.~50, no.~1, pp.
  89--99, 2019.

\bibitem{wu2017secure}
Y.~Wu and X.~He, ``Secure consensus control for multiagent systems with attacks
  and communication delays,'' \emph{IEEE/CAA Journal of Automatica Sinica},
  vol.~4, no.~1, pp. 136--142, 2017.

\bibitem{mesbahi2010graph}
M.~Mesbahi and M.~Egerstedt, \emph{Graph Theoretic Methods in Multiagent
  Networks}.\hskip 1em plus 0.5em minus 0.4em\relax Princeton University Press,
  2010.

\bibitem{das2011cooperative}
A.~Das and F.~L. Lewis, ``Cooperative adaptive control for synchronization of
  second-order systems with unknown nonlinearities,'' \emph{International
  Journal of Robust and Nonlinear Control}, vol.~21, no.~13, pp. 1509--1524,
  2011.

\bibitem{lewis2020neural}
F.~Lewis, S.~Jagannathan, and A.~Yesildirak, \emph{Neural Network Control of
  Robot Manipulators and Non-linear Systems}.\hskip 1em plus 0.5em minus
  0.4em\relax CRC press, 2020.

\bibitem{cui2015distributed}
B.~Cui, T.~Ma, F.~L. Lewis, C.~Zhao, Y.~Song, and C.~Feng, ``Distributed
  adaptive consensus control of heterogeneous multi-agent chaotic systems with
  unknown time delays,'' \emph{IET Control Theory \& Applications}, vol.~9,
  no.~16, pp. 2414--2422, 2015.

\bibitem{aryankia2020formation}
K.~Aryankia and R.~R. Selmic, ``Formation control and target tracking for a
  class of nonlinear multi-agent systems using neural networks,'' in \emph{2020
  European Control Conference (ECC)}.\hskip 1em plus 0.5em minus 0.4em\relax
  IEEE, 2020, pp. 160--165.

\bibitem{modares2019resilient}
H.~Modares, B.~Kiumarsi, F.~L. Lewis, F.~Ferrese, and A.~Davoudi, ``Resilient
  and robust synchronization of multiagent systems under attacks on sensors and
  actuators,'' \emph{IEEE Transactions on Cybernetics}, vol.~50, no.~3, pp.
  1240--1250, 2019.

\bibitem{sarangapani2018neural}
J.~Sarangapani, \emph{Neural Network Control of Nonlinear Discrete-Time
  Systems}.\hskip 1em plus 0.5em minus 0.4em\relax CRC press, 2018.

\bibitem{niu2019attack}
H.~Niu, C.~Bhowmick, and S.~Jagannathan, ``Attack detection and approximation
  in nonlinear networked control systems using neural networks,'' \emph{IEEE
  Transactions on Neural Networks and Learning Systems}, vol.~31, no.~1, pp.
  235--245, 2019.

\end{thebibliography}

\end{document}